\documentclass[runningheads]{llncs}

\newif\ifsubmit
\newif\ifextended

\submittrue
\extendedtrue

\usepackage[T1]{fontenc}
\usepackage{multirow}
\usepackage{graphicx}
\usepackage[pdfpagelabels=true]{hyperref}

\usepackage{xcolor}

\usepackage{amsfonts, amsmath, amssymb}
\usepackage{algorithm}
\usepackage{xspace}
\usepackage{algpseudocode}
\usepackage{url}
\usepackage{paralist}
\usepackage{ifthen}
\usepackage{tikz}
\usepackage{placeins}
\usetikzlibrary{quotes,arrows.meta}
\usetikzlibrary{quotes,arrows.meta}
\usetikzlibrary{shapes}
\usepackage{booktabs}
\usepackage{subcaption}
\usepackage{xfrac}

\newcommand{\eg}{e.g.,\@\xspace}

\newcommand{\colorpar}[3]{\colorbox{#1}{\parbox{#2}{#3}}}
\newcommand{\marginremark}[3]{\marginpar{\colorpar{#2}{3.5em}{\color{#1}#3}}}
\newcommand{\commentside}[2]{\marginpar{\color{#1}\tiny#2}}

\ifsubmit
  \newcommand{\TODO}[1]{}
  \newcommand{\REMARK}[1]{}
  \newcommand{\kvg}[1]{}
  \newcommand{\mv}[1]{}
  \newcommand{\tw}[1]{}
\else
  \newcommand{\TODO}[1]{\commentside{teal}{\textsc{Todo:} #1}}
  \newcommand{\REMARK}[1]{\commentside{teal}{\textsc{Remark:} #1}}
  \newcommand{\kvg}[1]{\marginremark{purple}{white}{\tiny{[KVG]~ #1}}}
  \newcommand{\mv}[1]{\marginremark{black}{yellow}{\tiny{[MV]~ #1}}}
  \newcommand{\tw}[1]{\marginremark{blue}{white}{\tiny{[TW]~ #1}}}
\fi

\newcommand{\Region}{\ensuremath{\mathcal{R}}}

\newcommand{\Rpos}{\ensuremath{\mathbb{R}_{> 0}}}
\newcommand{\PS}{\ensuremath{\textit{PS}}}
\newcommand{\MS}{\ensuremath{\textit{MS}}}
\newcommand{\MA}{\ensuremath{\mathcal{M}}}
\newcommand{\pMA}{\ensuremath{\mathcal{M}^{\mathcal{X}}}}
\newcommand{\dMA}{\ensuremath{\mathcal{M}^\delta}}
\newcommand{\dpMA}{\ensuremath{{\mathcal{M}^{\mathcal{X}}}^\delta}}

\newcommand{\MAval}[1]{\ensuremath{\mathcal{M}[#1]}}
\newcommand{\dMAlow}{\ensuremath{\underline{\mathcal{M}^\delta}}}
\newcommand{\dMAhigh}{\ensuremath{\overline{\mathcal{M}^\delta}}}
\newcommand{\Pmax}{\ensuremath{\textit{Pmax}}}
\newcommand{\Pmin}{\ensuremath{\textit{Pmin}}}

\newcommand{\tool}[1]{\textsc{#1}\xspace}
\newcommand{\storm}{\tool{Storm}}
\newcommand{\modest}{\tool{Modest}}
\newcommand{\stormpy}{\tool{Stormpy}}
\newcommand{\iscasmc}{\tool{iscasMC}}
\newcommand{\imca}{\tool{IMCA}}
\newcommand{\prismpsy}{\tool{PRISM-PSY}}
\newcommand{\prism}{\tool{PRISM}}
\newcommand{\qcomp}{\tool{QCOMP}}

\newcommand\smallpar[1]{%
\medskip\noindent\emph{#1}}

\newcommand{\para}[1]{\ifthenelse{\boolean{extended}}{\paragraph{#1}}{\smallpar{#1}}}

\def\orcidID#1{\smash{\href{http://orcid.org/#1}{\protect\raisebox{-1.25pt}{\protect\includegraphics{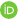}}}}}
\makeatother

\RequirePackage{graphicx}
\usepackage[firstpageonly=true,angle=0,vpos=.147\paperheight,hpos=.39\linewidth,vanchor=t,hanchor=l]{draftwatermark}
\SetWatermarkText{%
\raisebox{-3.5cm}
{\begin{minipage}{\linewidth}\centering\includegraphics[width=12.5mm] 
{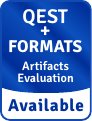}\hfill\includegraphics[width=12.5mm] 
{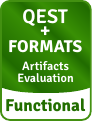}\end{minipage}}%
}

\begin{document}
\title{Verification of Parametric Markov Automata under Time-bounded Reachability}
\titlerunning{Verification of parametric Markov Automata}
%
\author{Kevin van de Glind\orcidID{0009-0004-0205-2346}%
\thanks{This work has been partially supported by NWO VENI Grant UNICORN (232.211).}
\and
Matthias Volk\orcidID{0000-0002-3810-4185}
\and
Tim~A.C.~Willemse\orcidID{0000-0003-3049-7962}}
%
%
\institute{Eindhoven University of Technology, Eindhoven, The Netherlands, \email{k.v.d.glind@tue.nl}}
\maketitle              

\begin{abstract}
Analysis of Markov models is of high importance for formal verification. Until now, analysis of Markov Automata required them to be fully specified, which is a considerable restriction as rates may be unknown or influenced by uncertainty of the environment. We introduce parametric Markov Automata (pMA) to capture this uncertainty with parametric transition functions. On these parametrized models, two different synthesis problems for time-bounded reachability properties are considered: I) Does there exist a valuation in the parameter space such that the instantiated model satisfies/violates the property, and II) given the parameter space, how can it be partitioned into satisfying and violating regions?

Our approach comprises two steps: I) The pMA is discretized to a parametric Markov decision process (pMDP), and II) through analysis of the pMDP, bounds are obtained for the reachability probability in the parametric MA. This approach solves the above problems up to a specified precision, as the accumulated error terms can be made arbitrarily small.

We implemented the approach using the \storm model checker.
Our experimental evaluation shows that the main performance bottlenecks originate from the discretization of the pMA.

\end{abstract}

\section{Introduction}

\emph{Markov Automata (MA)}~\cite{DBLP:conf/lics/EisentrautHZ10} are a compositional model for nondeterministic, continuous stochastically timed systems. They are able, as motivated in~\cite{Butkova_Fox_2019}, to serve as semantics for other reliability models such as Generalized Stochastic Petri Nets~\cite{ajmone1984class} and Dynamic Fault Trees~\cite{dugan2002dynamic}. Analysis of MA is supported by various tools, such as \storm~\cite{DBLP:journals/sttt/HenselJKQV22}, \modest~\cite{DBLP:conf/tacas/HartmannsH14} and \imca~\cite{DBLP:conf/nfm/GuckHKN12}. 

Until now, all algorithms for model checking MA required the model to be completely specified. In particular, each probabilistic and rate transition had to be precisely known prior to analysis. This can be an unrealistic assumption for real-world systems, as rates or probabilities may be subject to environmental changes or not precisely known.
In order to reason about such uncertainty, we introduce \emph{parametric Markov Automata~(pMA)}, where transitions can be parametrized.
This parametric extension of MA follows similar extensions of both discrete-time models~\cite{DBLP:journals/fmsd/JungesAHJKQV24} and continuous-time Markov chains~\cite{DBLP:conf/rtss/HanKM08,DBLP:journals/acta/CeskaDPKB17}.

Given a pMA and a region of parameter valuations, we solve two problems for time-bounded reachability properties. Firstly, \emph{satisfaction verification} checks whether the valuations in the region satisfy the property. Secondly, \emph{synthesis}~\cite{DBLP:journals/fmsd/JungesAHJKQV24} partitions the region into satisfying and violating subregions.
We solve both problems by a two-step approach. First, we discretize the pMA and obtain a \emph{parametric Markov decision process (pMDP)}. Second, after some transformations on the pMDP to make it suitable for further analysis, we perform \emph{parameter lifting (PL)}~\cite{DBLP:conf/atva/QuatmannD0JK16} on this discretized model. PL encodes the choice of parameter valuation as non-deterministic choices, yielding an MDP that allows analysis via off-the-shelf methods.
The proposed approach yields upper and lower bounds for the true reachability probability on the pMA, governed by three different error terms.
These bounds are then used for a classical threshold synthesis scheme~\cite{DBLP:journals/fmsd/JungesAHJKQV24,DBLP:journals/acta/CeskaDPKB17}.

\para{Related work.}
Similar models to MA include \emph{Continuous-Time MDPs~(CTMDPs)}~\cite{Miller_1968} and \emph{Interactive Markov Chains~(IMC)}~\cite{hermanns2002interactive}. The former allows for states with a combination of nondeterminism and timing, the latter already assumes a clear separation between these two types of states. In the literature, these became known as probabilistic and Markovian states. Markov Automata generalize both and introduced the possibility of immediate choice with a successor distribution in probabilistic states~\cite{DBLP:journals/eceasst/HatefiH12}. Several approaches have been proposed to analyze time-bounded reachability properties on MA: some operate directly on the continuous model~\cite{Butkova_Fox_2019,ashok2018continuous}, while others first discretize the model and then carry out the analysis using step-bounded queries~\cite{DBLP:journals/corr/GuckHHKT14,DBLP:journals/eceasst/HatefiH12,neuhausser2010model}.

Parametric extensions are natural for Markov models.
Model checking algorithms exist for discrete-time models~\cite{10.1007/978-3-540-31862-0_21,hahn2011probabilistic,DBLP:journals/fmsd/JungesAHJKQV24} supported through \eg \storm~\cite{DBLP:journals/sttt/HenselJKQV22} and \iscasmc~\cite{10.1007/978-3-319-06410-9_22}, probabilistic timed automata~\cite{DBLP:conf/qest/HartmannsKKS21} and continuous-time Markov chains~\cite{DBLP:conf/rtss/HanKM08,DBLP:journals/acta/CeskaDPKB17} through \prismpsy~\cite{DBLP:conf/tacas/CeskaPPBK16}.
A subset of parametric models are interval models, where each transition is assigned an interval of possible values instead. 
These are supported by \eg \prism~\cite{KNP11} or \storm~\cite{DBLP:journals/sttt/HenselJKQV22}.

\para{Contributions.}
This paper introduces parametric Markov automata. 
Our main contribution is an algorithm for time-bounded reachability on pMA. 
To perform analysis, we generalize the discretization result of~\cite{DBLP:journals/corr/GuckHHKT14} to pMA.
Furthermore, under common mild assumptions---non-zenoness of the pMA, multi-affine parametric functions and a weak global condition
---we transform the discretized pMA to an MDP, and solve both satisfaction verification and synthesis up to arbitrary precision.
We implemented our approach using \storm and demonstrate it on a variety of pMA models.

\para{Outline.}
Section~\ref{section:prelims} provides the necessary background on Markov Automata and discretization results on MA.
We introduce pMA in Section~\ref{section:parametricMA} and formally define the satisfaction verification and synthesis problems.
Section~\ref{section:results} covers the necessary transformations from pMA to a discretized model over which parameter lifting can be performed. Using PL, we solve the satisfaction verification and synthesis problems in Section~\ref{section:ParSynthStrategies}.
We experimentally evaluate our approach on a variety of benchmark models in Section~\ref{section:experiments} and conclude in Section~\ref{section:conclusion}.

\section{Preliminaries}
\label{section:prelims}
We briefly introduce Markov Automata and the discretization analysis technique~\cite{DBLP:journals/corr/GuckHHKT14} on which our approach is based.

\begin{definition}[Markov Automaton]\label{def:MarkovAutomaton}
    A \emph{Markov Automaton (MA)} is a tuple $\MA = (S, A, s_0, \rightarrow, \Rightarrow)$, where:
    \begin{itemize}
	    \item $S$ is a finite set of \emph{states}. 
	    \item $A$ is a finite set of \emph{actions}. 
	    \item $s_0\in S$ is the \emph{initial state}.
	    \item $\rightarrow\subseteq S\times A\times D(S)$ is the \emph{probabilistic} transition relation, where $D(S)$ is the set of discrete distributions over $S$.
	    \item $\Rightarrow\subseteq S\times\Rpos\times S$ is the \emph{rate} transition relation. The \emph{rate} of a transition is given by $
            R(s,s')=\sum_{ (s,r,s')\in \Rightarrow}r $.
    \end{itemize}
\end{definition}
Under the \emph{maximal progress assumption}---a common assumption for verification of MA~\cite{10.1007/978-3-030-30281-8_4}---we assume that probabilistic transitions happen without delay and therefore take precedence over rate transitions which happen instantaneously with probability zero. We can therefore partition the state space into \emph{probabilistic} and \emph{Markovian} states $S=\PS\sqcup\MS$.
The two transition relations can then be restricted to probabilistic and Markovian states, respectively: $\rightarrow\subseteq \PS\times A\times D(S)$ and $\Rightarrow\subseteq \MS\times\Rpos\times S$.
Figure~\ref{subfig:MAexampleInstantiated} depicts an example MA.

For a Markovian state $s \in \MS$ the \emph{sojourn time} of the state is governed by an \emph{exponential distribution} with rate $\lambda(s) = \sum_{s'\in S}R(s,s')$. Thus, the probability of leaving a state $s$ within $t$ units is given by $1-e^{-\lambda(s)t}$. Afterwards, we transition to successor state $s' \in S$ with probability $p(s,s') = \frac{R(s,s')}{\lambda(s)}$.
A probabilistic state $s \in \PS$ is left immediately by choosing an available action $\alpha\in A$ with $(s,\alpha,\mu)\in \rightarrow$, and transitions to state $s'$ with probability $p(s, s') = \mu(s')$.

An MA can therefore alternatively be defined as a tuple $\MA = (S,A,s_0,\lambda,p)$, where we set $\lambda(s)=+\infty$ for all probabilistic states $s\in\PS$, similar to~\cite{DBLP:journals/corr/GuckHHKT14}.

A \emph{scheduler} $\sigma$ resolves non-determinism on an MA $\MA$ by choosing available actions in probabilistic states, resulting in a deterministic model $\MA_\sigma$. Schedulers $\sigma:\PS\times \mathbb{R}_{\geq 0} \rightarrow A$ base this choice on the current state and the total time spent in the system. These belong to the class of total-time positional deterministic schedulers, denoted by $TP(\MA)$. Under a scheduler, one can define a probability measure on the MA; we refer to~\cite{DBLP:conf/fossacs/NeuhausserSK09} for the details.

MA are a superclass of \emph{Markov decision processes (MDPs)} and \emph{continuous-time Markov chains (CTMCs)}.
An MDP is a tuple $(S,A,s_0,p)$ and can be seen as an MA where $\MS=\emptyset$. Similarly, a CTMC can be viewed as an MA where $\PS=\emptyset$. For a more detailed description of both models, we refer to~\cite{DBLP:conf/lics/Katoen16}.

\para{Time-bounded reachability.}
In this paper, we focus the analysis on \emph{time-bounded reachability properties} of the form $\mathbb{P}[\diamond^{\leq t}G]\sim\theta$, computing: ``Is the maximal/mini\-mal probability to reach a set of goal states $G \subseteq S$ within $t$ time units greater/less than a given threshold $\theta$?''. For a given MA $\MA$ and scheduler $\sigma$, the \emph{time-bounded reachability probability} is denoted by $Pr^{\MA_\sigma}[\diamond^{\leq t} G]$. Comparing this value with a user-specified threshold $\theta\in~[0,1]$ yields a time-bounded reachability property. This can be formulated as a CSL model checking problem~\cite{DBLP:conf/cav/AzizSSB96}:
\begin{align*}
    \MA_\sigma\models \mathbb{P}[\diamond^{\leq t}G]\sim\theta \quad \textit{iff} \quad Pr^{\MA_\sigma}[\diamond^{\leq t} G]\sim\theta, \text{ for} \sim\in\{<, \leq, \geq, >\},\theta\in [0,1].
\end{align*}
For time-bounded reachability probabilities, total-time positional deterministic schedulers $TP(\MA)$ suffice~\cite{neuhausser2010model} and we define the maximal/minimal probability:
\begin{align*}
    &\Pmax[\MA, \diamond^{\leq t} G] = \sup_{\sigma\in TP(\MA)}Pr^{\MA_\sigma}[\diamond^{\leq t} G]\\
    &\Pmin[\MA, \diamond^{\leq t} G] = \inf_{\sigma\in TP(\MA)}Pr^{\MA_\sigma}[\diamond^{\leq t} G]
\end{align*}
\noindent Similarly to before, one can use $Pm[\MA, \diamond^{\leq t} G]\leq \theta$ for $m\in\{min,max\}$ to define time-bounded reachability properties.

\para{Discretization.}
Time-bounded reachability properties can be analysed using the \emph{discretization method}~\cite{DBLP:journals/corr/GuckHHKT14,DBLP:conf/atva/ButkovaHHK15,neuhausser2010model}, also called digitization~\cite{DBLP:journals/eceasst/HatefiH12}.
This transformation allows one to reason about the approximate behaviour of the MA by analyzing the time-discretized model, yielding a (discrete-time) MDP.
\begin{definition}[Discretized MA]\label{def:MAdiscretization}
	Let MA $\MA=(S,A,s_0,\rightarrow, \Rightarrow)$ and discret\-ization step $\delta>0$. The \emph{discretized MA} is an MDP $\dMA = (S,A, s_0,P')$, where:
    \begin{align*}
        P'(s,s') = \begin{cases}
            (1-e^{-\lambda(s)\delta})p(s,s'), &\text{if } s\neq s'\\
            (1-e^{-\lambda(s)\delta})p(s,s) + e^{-\lambda(s)\delta}, &\text{if } s=s'
        \end{cases}
    \end{align*}
\end{definition}
An example of this transformation is shown in Figure~\ref{subfig:MAexampleInstantiated}-\ref{subfig:MAexampleDiscretized}.
The central idea relies on the fact that, for small $\delta$, the probability of taking more than one Markovian step within time $\delta$ is small.
The error made by this discretization can be bounded~\cite{DBLP:journals/corr/GuckHHKT14}.
\begin{theorem}\label{thm:GuckDiscretization}\cite[Theorem 5.3]{DBLP:journals/corr/GuckHHKT14}
	Let $\MA$ be an MA, $G\subseteq S, \lambda = \max_{s\in \MS}\lambda(s)$, $t>0$ and fix $\delta>0$ such that $t/\delta\in \mathbb{N}$.
    Then the error can be bounded as follows:
    \begin{align*}
        \Pmax[\dMA, \diamond^{\leq t/\delta}G]
        \leq \Pmax[\MA,\diamond^{\leq t}G] 
        \leq \Pmax[\dMA,\diamond^{\leq t/\delta}G] + \epsilon_\delta,
    \end{align*}
    where $\epsilon_\delta = 1-e^{-\lambda t}(1+\lambda\delta)^{t/\delta}$ denotes the \emph{discretization error}.
    The same bounds hold for $Pmin[\cdot]$.
\end{theorem}
Step-bounded reachability analysis on the discretized MA requires an adapted value iteration scheme, as probabilistic states are left immediately~\cite{DBLP:journals/eceasst/HatefiH12,DBLP:journals/corr/GuckHHKT14}.

\begin{figure}[!hbt]
    \begin{subfigure}[!h]{0.30\textwidth}
    \centering
        \begin{tikzpicture}[->]
\node[ellipse,draw] (s0) at (1.1,2) {$s_0$};
\node[ellipse,draw] (s1) at (2.8,3) {$s_1$};
\node[ellipse,draw] (s2) at (2.8,1) {$s_2$};
\node[ellipse,draw] (sG) at (4.0,1) {$s_3$};
\node[ellipse,draw] (sE) at (4.0,3) {$s_4$};

 \path[draw=olive,line width=1pt] (s0) edge             node[above] {\tiny$a$} (1.7,1.5);
 \fill (1.7,1.5) circle (1.5pt);
 \path (1.7,1.5) edge             node[right, very near start, xshift=-2pt] {\tiny$0.7$} (s1);
 \path (1.7,1.5) edge            node[below, near start] {\tiny$0.3$} (s2);
 \path[draw=olive,line width=1pt] (s0) edge             node[above, near start] {\tiny$b$} (1.7,2.5);
 \fill (1.7,2.5) circle (1.5pt);
\path (1.7,2.5) edge              node[above, right, very near start] {\tiny$p$} (s2);
\path (1.7,2.5) edge              node[above, near start] {\tiny$1{-}p$} (s1);

\path[draw=teal,line width=1pt] (s1) edge              node[above, near start] {\tiny$1$} (sE);
\path[draw=teal,line width=1pt] (s1) edge              node[above, near start] {\tiny$1$} (sG);
\path[draw=teal,line width=1pt] (s2) edge              node[below, near start] {\tiny$4$} (sE);
\path[draw=teal,line width=1pt] (s2) edge              node[below, near start] {\tiny$q$} (sG);
\path[draw=teal,line width=1pt] (sE) edge [loop above] node {\tiny $1$} (sE);
\path[draw=teal,line width=1pt] (sG) edge [loop above] node {\tiny $1$} (sG);

\end{tikzpicture}
        \caption{Parametric MA}
        \label{subfig:MAexampleParametric}
    \end{subfigure}
    \begin{subfigure}[!h]{0.30\textwidth}
    \centering
        \begin{tikzpicture}[->]
\node[ellipse,draw] (s0) at (1.1,2) {$s_0$};
\node[ellipse,draw] (s1) at (2.8,3) {$s_1$};
\node[ellipse,draw] (s2) at (2.8,1) {$s_2$};
\node[ellipse,draw] (sG) at (4.0,1) {$s_3$};
\node[ellipse,draw] (sE) at (4.0,3) {$s_4$};

 \path (s0) edge             node[above] {\tiny$a$} (1.7,1.5);
 \fill (1.7,1.5) circle (1.5pt);
 \path (1.7,1.5) edge             node[right, very near start, xshift=-2pt] {\tiny$0.7$} (s1);
 \path (1.7,1.5) edge            node[below, near start] {\tiny$0.3$} (s2);
 \path (s0) edge             node[above, near start] {\tiny$b$} (1.7,2.5);
 \fill (1.7,2.5) circle (1.5pt);
\path (1.7,2.5) edge              node[above, right, very near start] {\tiny$0.6$} (s2);
\path (1.7,2.5) edge              node[above, near start] {\tiny$0.4$} (s1);

\path (s1) edge              node[above, near start] {\tiny$1$} (sE);
\path (s1) edge              node[above, near start] {\tiny$1$} (sG);
\path (s2) edge              node[below, near start] {\tiny$4$} (sE);
\path (s2) edge              node[below, near start] {\tiny$7$} (sG);
\path (sE) edge [loop above] node {\tiny $1$} (sE);
\path (sG) edge [loop above] node {\tiny $1$} (sG);

\end{tikzpicture}
        \caption{Instantiated MA}
        \label{subfig:MAexampleInstantiated}
    \end{subfigure}
    \begin{subfigure}[!h]{0.38\textwidth}
    \centering
        \begin{tikzpicture}[->]
\node[ellipse,draw] (s0) at (1.1,2) {$s_0$};
\node[ellipse,draw] (s1) at (2.8,3) {$s_1$};
\node[ellipse,draw] (s2) at (2.8,1) {$s_2$};
\node[ellipse,draw] (sG) at (5.2,1) {$s_3$};
\node[ellipse,draw] (sE) at (5.2,3) {$s_4$};

 \path (s0) edge             node[above] {\tiny$a$} (1.7,1.5);
 \fill (1.7,1.5) circle (1.5pt);
 \path (1.7,1.5) edge             node[right, very near start, xshift=-2pt] {\tiny$0.7$} (s1);
 \path (1.7,1.5) edge            node[below, near start] {\tiny$0.3$} (s2);
 \path (s0) edge             node[above, near start] {\tiny$b$} (1.7,2.5);
 \fill (1.7,2.5) circle (1.5pt);
\path (1.7,2.5) edge              node[above, right, very near start] {\tiny$0.6$} (s2);
\path (1.7,2.5) edge              node[above, near start] {\tiny$0.4$} (s1);

\path (s1) edge [loop above] node [above, yshift=-2pt] {\tiny $e^{-2\delta}$} (s1);
\path (s2) edge [loop above] node [above, yshift=-2pt] {\tiny $e^{-11\delta}$} (s2);
\path (s1) edge              node[above] {\tiny$\frac{1-e^{-2\delta}}{2}$} (sE);
\path (s1) edge              node[above, right, xshift=-23pt, yshift=18pt] {\tiny$\frac{1-e^{-2\delta}}{2}$} (sG);
\path (s2) edge              node[below, right, xshift=-25pt, yshift=-19pt] {\tiny$\frac{4(1-e^{-11\delta})}{11}$} (sE);
\path (s2) edge              node[below, yshift=2pt] {\tiny$\frac{7(1-e^{-11\delta})}{11}$} (sG);
\path (sE) edge [loop above] node {\tiny $1$} (sE);
\path (sG) edge [loop above] node {\tiny $1$} (sG);

\end{tikzpicture}
        \caption{Discretized MA}
        \label{subfig:MAexampleDiscretized}
    \end{subfigure}
    \caption{(a) Parametric MA with nondeterministic transitions in {\color{olive}olive} and rate transitions in {\color{teal}teal}, (b) instantiated MA for instantiation $p \mapsto 0.6, q \mapsto 7$, and (c) its discretization using discretization step $\delta$.}
    \label{fig:MAdiscretization}
\end{figure}
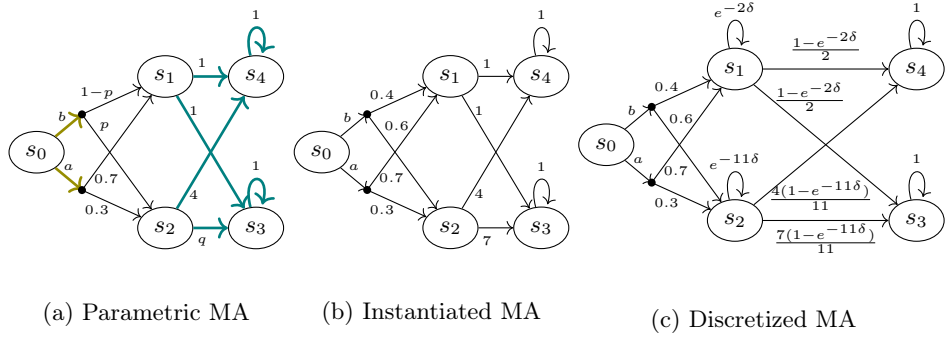

\section{Parametric Markov Automata}\label{section:parametricMA}
In this section, we define \emph{parametric MA} which extend MA with parametric rate- and probabilistic functions.
We lift the discretization result of Theorem~\ref{thm:GuckDiscretization} to this parametric model in Lemma~\ref{lemma:parametricMADiscretization}. Last, we provide the formal problem definition, defining the satisfaction verification and synthesis problems.

\para{Parametric Markov automata.}
We first introduce the necessary notation for parameters.
Let $\mathcal{X} = \{x_1,...,x_N\}$ be a finite set of \emph{parameters}. A \emph{(parameter) instantiation} is a mapping $\mathcal{X}\rightarrow \mathbb{R}$. A parametrized successor distribution is denoted by $D^\mathcal{X}: S\rightarrow [0,1]$. A \emph{parameter region} $\Region$ is a closed and bounded subset, in the topological sense as described in for example~\cite{garling2014course}, $\Region\subset\mathbb{R}_{\geq 0}^{|\mathcal{X}|}$ with
any $v\in \Region$ called a \emph{valuation}. Moreover, if $\Region$ is a polytope, we write $V_\Region$ to denote its vertices. A polynomial over $\mathcal{X}$ is a weighted sum over monomials $x_1^{a_1} ... x_n^{a_n}$ for $a_1,...,a_n\in\mathbb{N}$ and it is called \emph{multi-affine} if for each monomial $a_1,...,a_n\leq 1$.
The set of polynomials over $\mathcal{X}$ with real coefficients is denoted by $\mathbb{R}[\mathcal{X}]$. A polynomial $f\in \mathbb{R}[\mathcal{X}]$ is \emph{monotone increasing/decreasing} on $\Region$ with respect to $x\in \mathcal{X}$ if, respectively, $\frac{d}{dx}f|_\Region\geq 0 $ or $\frac{d}{dx}f|_\Region\leq 0$.
A polynomial $f$ is called \emph{monotone} on $\Region$ if it is monotone with respect to each parameter. Multi-affine functions are monotone on any region $\Region$. 

\begin{definition}[Parametric Markov Automaton]\label{def:pMA}
    A \emph{parametric Markov Automaton (pMA)} over $\mathcal{X}$ is a tuple $\pMA=(S, A,s_0,\rightarrow,\Rightarrow)$, where states $S$, actions $A$ and initial state $s_0$ are as in Definition~\ref{def:MarkovAutomaton}, and the two transition relations $\rightarrow \subseteq PS\times A \times D^{\mathcal{X}}(S)$ and $\Rightarrow \subseteq MS\times \mathbb{R}[\mathcal{X}] \times S$ are parametrized over $\mathcal{X}$. 
\end{definition}

Similarly to pMDPs and pCTMCs~\cite{DBLP:conf/rtss/HanKM08,DBLP:journals/fmsd/JungesAHJKQV24}, pMAs can be instantiated. An example of a pMA and an instantiation can be seen in Figure~\ref{subfig:MAexampleParametric}-\ref{subfig:MAexampleInstantiated}.

\begin{definition}[Instantiated MA]\label{def:pMAinstantiation}
        For any pMA $\pMA$ and valuation $v\in\mathbb{R}^{|\mathcal{X}|}$, we obtain the \emph{instantiated MA} $\pMA[v]$ by applying the instantiation $x_i\mapsto v_i$.
\end{definition}
We use the same notation for an instantiation of a single rate or successor distribution with valuation $v$ as for instantiating a pMA.
A valuation $v$ is called \emph{well-defined} if the resulting instantiated MA $\pMA[v]$ satisfies Definition~\ref{def:MarkovAutomaton}.
A well-defined valuation is \emph{graph-preserving} if the following conditions are met:
\begin{alignat*}{2}
    \forall &s\in \MS,s'\in S,v\in\Region: &&\quad R(s,s') \not\equiv 0 \rightarrow R[v](s,s')\neq 0\\
    \forall &s\in \PS,s'\in S, \alpha\in A,v\in\Region: &&\quad P(s,\alpha,s') \not\equiv 0 \rightarrow P[v](s,\alpha,s')\neq 0
\end{alignat*}
In the remainder we assume each valuation is well-defined and graph-preserving.

A pMA is called \emph{non-zeno} if, for every valuation, the instantiated MA does not have a \emph{zeno cycle}, which is a cycle of probabilistic states one does not exit almost surely. For further details on zenoness of MA, we refer to~\cite{DBLP:journals/eceasst/HatefiH12,DBLP:journals/corr/GuckHHKT14}.
We assume each pMA to be non-zeno, a common assumption on MA in the literature~\cite{DBLP:conf/tacas/ButkovaWH17,DBLP:journals/eceasst/HatefiH12}. Moreover, in practice, non-zenoness, and the stronger occurrence of having any cycles of probabilistic states at all, is extremely rare~\cite{DBLP:conf/tacas/ButkovaWH17} and existing MA models, found for example in the \qcomp benchmark suite~\cite{DBLP:conf/tacas/HartmannsKPQR19}, do not have, to our knowledge, cycles of probabilistic states.

In order to determine the exact extremal values of parametric functions, these are restricted to multi-affine polynomials. This is in line with the restriction on parametric rate functions in~\cite{DBLP:journals/acta/CeskaDPKB17}. For the sake of presentation, we restrict attention to hyper-rectangular parameter spaces, a standard assumption in literature~\cite{DBLP:journals/acta/CeskaDPKB17,DBLP:journals/fmsd/JungesAHJKQV24}. However, our methods extend to finite unions of hyper-rectangles. 

\para{Lifting discretization.}
In order to solve time-bounded reachability queries on pMA, we generalize the discretization result of Theorem~\ref{thm:GuckDiscretization}.
The discretization of pMA is analogous to Definition~\ref{def:MAdiscretization}, but instead yields a parametric MDP. The generalization is shown in Lemma~\ref{lemma:parametricMADiscretization}, where, for any two parametric models $\MA, \MA'$ over the same parameters, we write $Pmax[\MA, \phi]\leq Pmax[\MA', \phi]$ if for each valuation $v\in \Region $ we have $Pmax[\MAval{v}, \phi]\leq Pmax[\MA'[v], \phi]$. 

\begin{lemma}\label{lemma:parametricMADiscretization}
    Let $\pMA$ be a pMA over a region $\Region$ and $t>0$. Fix $\delta > 0$ as discretization step such that $t/\delta\in\mathbb{N}$.
    We get the following:
    \begin{align*}
        Pmax[\dpMA, \diamond^{\leq t/\delta}G] \leq Pmax[\pMA, \diamond^{\leq t}G] \leq& Pmax[\dpMA, \diamond^{\leq t/\delta}G] \\&+ 1- \inf_{v\in \Region}[e^{-\lambda[v]t}(1+\lambda[v]\delta)^{t/\delta}] 
    \end{align*}
    The same results hold for $Pmin[\cdot]$.
\end{lemma}

\begin{proof}
    Consider any valuation $v\in \Region$. By Theorem~\ref{thm:GuckDiscretization} the following bound is obtained: $
         Pmax[\dpMA[v],\diamond^{\leq t/\delta}G] \leq Pmax[\pMA[v],\diamond^{\leq t}G]
$.
 Therefore, we immediately prove our lower bound. Similarly, we get: 
    \begin{align*}
        Pmax[\pMA[v],\diamond^{\leq t}G] &\leq Pmax[\dpMA[v],\diamond^{\leq t/\delta}G] + 1- e^{-\lambda[v]t}(1+\lambda[v]\delta)^{t/\delta} \\
        &\leq Pmax[\dpMA[v],\diamond^{\leq t/\delta}G] + 1- \inf_{w\in R}[e^{-\lambda[w]t}(1+\lambda[w]\delta)^{t/\delta}]
    \end{align*}
    The proof for $Pmin[\cdot]$ is analogous. \qed
\end{proof}

The infimum in the error bound requires a maximization of the maximum sojourn rate. From this point onward, for readability, we will denote pMA and its corresponding discretization by $\MA$ and $\dMA$, respectively, whenever it is clear from the context that we are referring to the parametric model. In such cases, the discretization error will be written as $\epsilon_\delta = 1- \inf_{v\in \Region}[e^{-\lambda[v]t}(1+\lambda[v]\delta)^{t/\delta}]$.

\para{Problem statement.}
Now that pMA are formally defined, we provide a formal problem definition.
In contrast to MA, the analysis of pMA also needs to reason about the parameter valuations.
We solve two main problems for time-bounded reachability properties: (1)~\emph{satisfaction verification} checks whether there exists a valuation in the region which satisfies/violates the given property, and (2)~\emph{synthesis} partitions the parameter space into satisfying and violating subregions.

Four different satisfaction relations are considered, distinguishing cases where the \emph{player}---making the non-deterministic choices---and \emph{nature}---assigning the parameter values---can either be cooperative or adversarial, similar to Definition~14 in~\cite{DBLP:journals/fmsd/JungesAHJKQV24} for discrete-time models. 

\begin{definition}[Satisfaction relations]\label{def:satisfactionRelations}
    For a pMA $\MA$, region $\Region$ and time-bounded reachability property $\phi = \mathbb{P}[\diamond^{\leq t} G]~\sim\theta$, we define the following satisfaction relations:
    \begin{align*}
        \text{(Feasibility)}\quad&\MA,\Region \models_{f} \phi\quad\text{iff}\quad\exists \sigma\in TP(\MA)~\exists v\in \Region: \MA_\sigma[v]\models \phi\\
        \text{(Angelic-Aid)}\quad&\MA,\Region \models_{a} \phi\quad\text{iff}\quad\forall \sigma\in TP(\MA)~\exists v\in \Region: \MA_\sigma[v]\models \phi\\
        \text{(Robust)}\quad&\MA,\Region \models_{r} \phi\quad\text{iff}\quad\exists \sigma\in TP(\MA)~\forall v\in \Region: \MA_\sigma[v]\models \phi\\
        \text{(Demonic)}\quad&\MA,\Region \models_{d} \phi\quad\text{iff}\quad\forall \sigma\in TP(\MA)~\forall v\in \Region: \MA_\sigma[v]\models \phi
    \end{align*}
\end{definition}

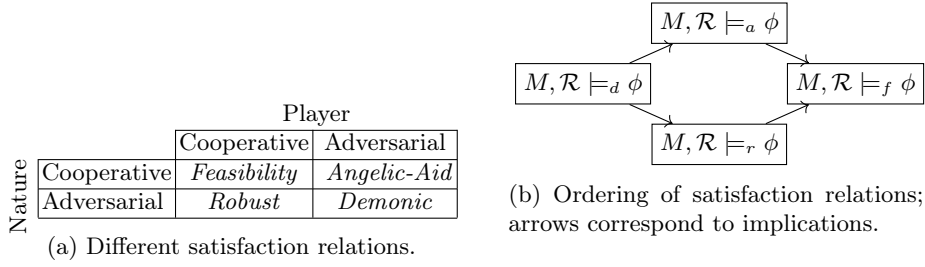
\begin{figure}[tb]
    \centering
\begin{subtable}[t]{0.50\textwidth}
    \centering
\begin{tabular}{l|l|c|c|}
\multicolumn{2}{c}{}&\multicolumn{2}{c}{Player}\\
\cline{3-4}
\multicolumn{2}{c|}{}&Cooperative&Adversarial\\
\cline{2-4}
\multirow{2}{*}{\rotatebox{90}{Nature} }& Cooperative & \textit{Feasibility} &\textit{ Angelic-Aid} \\
\cline{2-4}
& Adversarial & \textit{Robust} & \textit{Demonic} \\
\cline{2-4}
\end{tabular}
    \caption{Different satisfaction relations.}
    \label{table:satisfactionRelations}
\end{subtable}
\hfill
\begin{subfigure}[t]{0.45\textwidth}
    \hspace{0cm}  
    \begin{tikzpicture}[->]
    \node[draw] (Feasibility) at (3.6,0) {$M,\Region \models_{f} \phi$};
    \node[draw] at (1.8,0.8) (angel) {$M,\Region \models_{a} \phi$};
    \node[draw] (robust) at (1.8,-0.8) {$M,\Region \models_{r} \phi$};
    \node[draw] (demonic) at (0,0) {$M,\Region \models_{d} \phi$};
    \path (demonic) edge (robust);
    \path (demonic) edge (angel);
    \path (angel) edge (Feasibility);
    \path (robust) edge (Feasibility);
    \end{tikzpicture}
    \caption{Ordering of satisfaction relations; arrows correspond to implications.}
    \label{fig:strategyRelations}
\end{subfigure}
    \caption{Satisfaction relations.}
\end{figure}

The bivariate table in Table~\ref{table:satisfactionRelations} provides an intuition of the satisfaction relations and whether the player/nature has a cooperative or adversarial attitude towards satisfying the property. Moreover, there is a clear ordering on these satisfaction relations, as shown in Figure~\ref{fig:strategyRelations}.

We now define the two main problems.
Given a pMA $\MA$, region $\Region$, time-bounded reachability property $\phi$ and satisfaction relation $\clubsuit\in \{f,a,r,d\}$, the \emph{satisfaction verification} problem is to decide whether $\MA,\Region \models_{\clubsuit} \phi$. The verification problems for discrete-time models considered in~\cite{DBLP:journals/fmsd/JungesAHJKQV24} are the robust and demonic cases of this problem. The \emph{synthesis} problem is to find a partitioning $\Region=T\sqcup F$ such that $\forall{T'}\subseteq{T}:~\MA,T' \models_{\clubsuit}~\phi$ and $\forall{F'}\subseteq{F}:~\MA,F' \not\models_{\clubsuit}\phi$. This definition is equivalent to the approximate synthesis problem in~\cite{DBLP:journals/fmsd/JungesAHJKQV24}. As for each valuation in $T,F$ the property must be satisfied/violated, feasibility and robust satisfaction relations yield the same result, and so do angelic-aid and demonic.

\section{Model transformations}\label{section:results}

Using Lemma~\ref{lemma:parametricMADiscretization}, one could try to naively perform parameter synthesis on the discretized pMA. However, the discretized pMA requires an adapted value-iteration scheme~\cite{DBLP:journals/eceasst/HatefiH12}, which combines unbounded reachability between PS states with time-bounded reachability in MS.
Instead we transform the discretized pMA as outlined in Figure~\ref{fig:process}, which allows us to apply existing analysis techniques for step-bounded reachability on pMDP.
In the remainder of this section, we discuss these transformations in more detail.

Section~\ref{ss:CollapsePS} focuses on transforming and discretizing the pMA into a collapsed discretized pMA, where the instantaneous transitions from PS states are collapsed into Markovian states, using state-elimination approaches such as~\cite{10.1007/978-3-319-47677-3_6}.
The adapted reachability scheme of the original model then yields the same result as step-bounded reachability on the transformed model.
We could then perform parameter synthesis on the resulting pMDP.
However, discretizing the pMA is expensive and would need to be performed for each subregion.
Instead, Section~\ref{ss:DependentRates} introduces a method to obtain approximate discretizations $\dMAlow,\dMAhigh$ which include the discretization as a parameter in the discretized pMA, allowing the transformations covered in this section to be a preprocessing step.

After the preprocessing, we instantiate the discretization for each parameter region and perform the verification.
The verification either terminates with a conclusive result, or we need to further partition the region and check each subregion individually.
This process, including how to take approximation errors into account, is described in Section~\ref{section:ParSynthStrategies}.

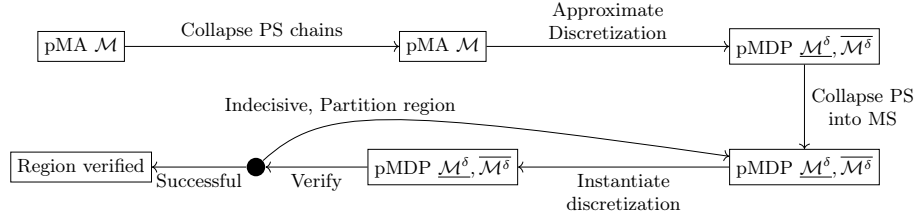
\begin{figure}[tb]
    \centering\resizebox{\textwidth}{!}{
    \begin{tikzpicture}[->]
    \node[draw] (ma) at (0,2) {pMA $\MA$};
    \node[draw] (ma_collapsed) at (6,2) {pMA $\MA$};
    \node[draw] (approx_discretization) at (12,2) {pMDP $\underline{\dMA},\overline{\dMA}$};
    \node[draw] (approx_discretization_collapsed) at (12,0) {pMDP $\underline{\dMA},\overline{\dMA}$};
    \node[draw] (approx_discretization_collapsed_instantiated) at (6,0) {pMDP $\underline{\dMA},\overline{\dMA}$};
    \node[draw] (verified) at (0,0) {Region verified};
    \node[circle,fill=black,minimum size=1pt] (decision) at (2.9,0) {};
    \path (ma) edge["Collapse PS chains"] (ma_collapsed);
    \path (ma_collapsed) edge["Approximate\\ Discretization",align=center](approx_discretization);
    \path (approx_discretization) edge[" Collapse PS \\ into MS",align=center](approx_discretization_collapsed);
    \path (approx_discretization_collapsed) edge [ " Instantiate\\discretization",align=center] (approx_discretization_collapsed_instantiated);
    \path (approx_discretization_collapsed_instantiated) edge[ "Verify"] (decision);
    \path (decision) edge[" Successful"] (verified);
    \draw[->] (decision) .. controls (4,1) and (5,1) .. (approx_discretization_collapsed);
    \node at (4.3,1) { Indecisive, Partition region};
\end{tikzpicture}}
    \caption{Outline of the process for different synthesis problems.}
    \label{fig:process}
\end{figure}

\subsection{Collapse of Probabilistic states}\label{ss:CollapsePS}
The original adapted VI scheme~\cite{DBLP:journals/eceasst/HatefiH12} only counts steps in the MDP if they are taken in states corresponding to Markovian states in the MA. Therefore, we propose a method to eliminate $\PS$ using state elimination~\cite{10.1007/978-3-319-47677-3_6}, such that all steps in the discretized model will be taken in $\MS$ and the alternative value-iteration scheme reduces to a step-bounded reachability query.

As time does not pass in a chain of probabilistic states and we assume non-zenoness, one can fix the policy in the first state of the chain and immediately get the distribution over Markovian states by eliminating the intermediate probabilistic states. Thus, a chain of $\PS$ can be reduced to a single state. In essence, this approach follows a similar idea to the MA to SAMA transformation in~\cite{GrosMScThesis}. This transformation naturally comes at the cost of an exponential blow-up in the number of available actions for the given chain. However, in practice, chains of PS tend to be short, which controls this growth of available actions.

To eliminate the remaining $\PS$, the pMA needs to be discretized, as pMA have a clear distinction between states with stochastic timing and non-deterministic choice. After discretization, we can collapse the PS into MS using well known elimination procedures~\cite{hahn2011probabilistic}.
\ifthenelse{\boolean{extended}}{%
	Appendix~\ref{appendix:algorithm} provides the details of this algorithm.
}{%
	Details of this algorithm can be found in~\cite{extended}.
}%
Correctness of this approach immediately follows from the adapted value-iteration scheme in~\cite{DBLP:journals/eceasst/HatefiH12}.
The cost of this transformation is again an exponential increase in the total number of actions with respect to the outdegree of Markovian states, which tends to be small in practice.

\subsection{Approximate discretization}\label{ss:DependentRates}

Throughout the discretization of the pMA, in accordance with Definition~\ref{def:MAdiscretization}, the parametric functions on the resulting pMDP will contain an exponential term. However, current methods~\cite{DBLP:journals/fmsd/JungesAHJKQV24,10.1007/978-3-642-20398-5_12} only allow for rational functions.
We therefore need to represent the exponential via rational functions. We briefly describe two methods: I) \emph{Minimizing/Maximizing the sojourn rates} to get the respective minimum/maximum reachability results, or II) \emph{approximating the exponential} by a partial Taylor expansion~\cite{stewart2021calculus}. Both of these result in over- and underapproximations of the final reachability probabilities. 

\para{Minimizing/Maximizing sojourn rates.}
The first method to obtain an over- and underapproximation of the reachability probability is to maximize/minimize the sojourn rates for each state independently from the probabilistic transitions of the same state. Clearly, maximizing these rates results in a higher probability of reaching the goal state within the set time bound, as states are left faster and the goal state can be reached earlier. In practice, this means that we set the sojourn rates of each state to $\max_{v\in V_\Region}\lambda[v](s)$ or $\min_{v\in V_\Region}\lambda[v](s)$.
We use the assumption of multi-affine parametric rate functions here to guarantee that the extrema are attained at the vertices of the region. The main drawback of this approach is that intra-state dependencies between sojourn rate and successor distribution are lost, potentially leading to severe approximation errors. 

\para{Approximating the exponential function.}
To maintain intra-state dependencies, we approximate the exponential function by a partial Taylor expansion:
\begin{align}\label{eq:expApprox}
    \forall N\in\mathbb{N}:\quad\sum_{k=0}^{2N+1}\frac{(-x)^k}{k!} \leq e^{-x} \leq \sum_{k=0}^{2N}\frac{(-x)^k}{k!}
\end{align}
The bounds in Equation~\ref{eq:expApprox} are directly implied by the alternating series remainder theorem~\cite{stewart2021calculus}.
Based on this approximation, we introduce the \emph{approximate discretized Markov Automaton}.

\begin{definition}[Approximate discretized pMA]\label{def:approxDiscretization}
            Given discretization step ${\delta>0}$, $N\in \mathbb{N}_{>0}$, and pMA $\MA=(S,A, s_0,\rightarrow, \Rightarrow)$.
            The \emph{approximate discretized parametric Markov Automaton} is given by $\dMA_N = (S,A,s_0, P')$, where:
    \begin{align*}
        &P'(s,s') = \begin{cases}
            (1-\sum_{k=0}^{N}\frac{(-\lambda(s))^k}{k!})p(s,s'), &\textit{ if } s\neq s'\\
            (1-\sum_{k=0}^{N}\frac{(-\lambda(s))^k}{k!})p(s,s') + \sum_{k=0}^{N}\frac{(-\lambda(s))^k}{k!}, &\textit{ otherwise.}\\
        \end{cases}
    \end{align*}
\end{definition}

 Using Theorem~\ref{thm:approxDiscretizationError}, given below, the approximate discretization is guaranteed to bound the reachability probability of the pMA, and the error of this additional approximation is also bounded in terms of the discretized pMA.

\begin{theorem}\label{thm:approxDiscretizationError}
        Given pMA $\MA=(S,A,s_0, \rightarrow, \Rightarrow), G\subseteq S,t>0 $. Fix discretization step $\delta>0$ so $t/\delta\in\mathbb{N}$ and pick $K\in\mathbb{N}_{>0}$. Let $\dMAlow = \dMA_{2K}$, $\dMAhigh = \dMA_{2K+1}$ and $\lambda =\sup_{v\in\Region}\lambda[v]$. Then we have that: 
        \begin{align*}
        &Pmax[\dMA,\diamond^{\leq t/\delta}G] - \epsilon^{2K}_{t/\delta} \leq Pmax[\dMAlow,\diamond^{\leq t/\delta}G]\\ 
        \leq &Pmax[\MA,\diamond^{\leq t}G]\\ 
        \leq &Pmax[\dMAhigh,\diamond^{\leq t/\delta}G] + \epsilon_\delta \leq Pmax[\dMA,\diamond^{\leq t/\delta}G] + \epsilon^{2K+1}_{t/\delta} + \epsilon_\delta,  
        \end{align*}
        where $\epsilon^M_{t/\delta} = (1+\frac{(\lambda\delta)^M}{M!})^{t/\delta}-1$ and $\epsilon_\delta$ is the discretization error as per Lemma~\ref{lemma:parametricMADiscretization}.
    Moreover, the same result holds for $Pmin[\cdot]$. 
\end{theorem}
 The proof of Theorem~\ref{thm:approxDiscretizationError} can be found \ifthenelse{\boolean{extended}}{%
	in Appendix~\ref{appendix:approximateDiscretization}.
}{%
	in the extended version of this paper~\cite{extended}.
}

\begin{figure}[tb]
    \centering
    \includegraphics[width=0.75\linewidth]{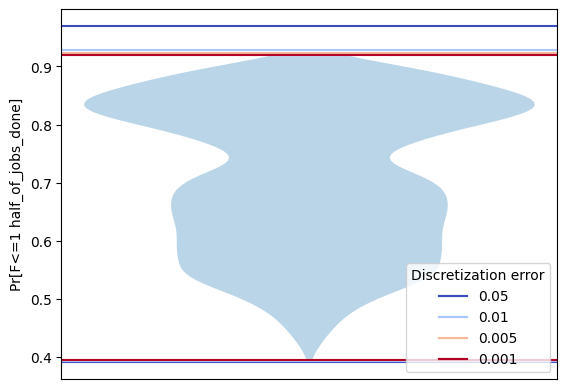}
    \caption{Violinplot of probabilities for $10^4$ pMA instantiations on Jobs benchmark with $\Region = [1,3]^2$; horizontal lines are the approximate discretization bounds.}
\label{fig:approximateBoundsVsRandomInstantiation}
\end{figure}

The discretization error $\epsilon_\delta$ dominates the error made by the approximate discretization.
This can be seen in Figure~\ref{fig:approximateBoundsVsRandomInstantiation}, where the resulting probabilities for pMA instantiations are plotted against the bounds obtained using the approximate discretization with $N$ equal to $1,2$ for $\dMAhigh,\dMAlow$ respectively.

\section{Solving satisfaction verification and synthesis problems}\label{section:ParSynthStrategies}
Section~\ref{section:results} focused on the necessary transformations to make step-bounded reachability on the discretized MA possible.
This section focuses on the verification of regions on the resulting MDP using parameter lifting~\cite{DBLP:conf/atva/QuatmannD0JK16}.

\subsection{Parameter lifting}\label{ss:ParameterLifting}
\emph{Parameter Lifting (PL)} is a technique which abstracts inter-state dependencies away with a transformation from pMDP to parameter-free \emph{stochastic games (SG)}~\cite{DBLP:conf/atva/QuatmannD0JK16}. The extremal values over the parametric transition functions are precomputed, and nature, as a second player, chooses an extremal parameter value which, depending on the satisfaction relation, minimizes/maximizes the time-bounded reachability probability~\cite{DBLP:journals/fmsd/JungesAHJKQV24}.
This relaxation introduces another error which becomes smaller when the absolute difference between extrema decreases. 

As PL assumes that extrema are attained at the vertices of the region~\cite{DBLP:conf/atva/QuatmannD0JK16}, we must ensure that this holds throughout the transformations from Section~\ref{section:results}. Therefore, we assume that no two states in a chain of probabilistic states share a parametric function defined over the same variable. Moreover, to cover the collapse of PS in MS, we assume that the sets of parameters over rate- and probabilistic transition functions are disjoint. These assumptions can be enforced, if needed, by introducing auxiliary parameters. We show \ifthenelse{\boolean{extended}}{%
	in Appendix~\ref{appendix:extremaPL}
}{%
	in~\cite{extended}
}that I) by adding auxiliary parameters, one does not lose the guarantee of reasoning up to arbitrary precision over the original parameter space and II) these assumptions suffice to guarantee correctness of PL. 
None of the models we have come across from the literature are negatively impacted by these assumptions.
Moreover, the same assumptions are used for pCTMC~\cite{DBLP:journals/acta/CeskaDPKB17}, motivating that these conditions are weak enough to model realistic scenarios.    

\subsection{Solving the safisfaction verification and synthesis problems}\label{ss:solvingSynthesisProblems}
The problems can now be solved by comparing the bounds obtained through PL with the threshold $\theta$, using a well-known refinement scheme~\cite{DBLP:journals/acta/CeskaDPKB17,DBLP:journals/fmsd/JungesAHJKQV24,vcevska2019shepherding,DBLP:journals/corr/BraitlingFHWBH14}:
{\par\nopagebreak\small\noindent\ignorespaces%
\small\begin{align*}
    &M,\Region \models_{\clubsuit} \mathbb{P}[\diamond^{\leq t}G]\leq \theta =\begin{cases}
        \text{Satisfied, if }upper_\clubsuit \leq \theta,\\
        \text{Violated, if }lower_\clubsuit > \theta,\\
        \text{Inconclusive and needs further splitting}, \textit{Otherwise.}\\
    \end{cases}
\end{align*}}%

The bounds $upper_\clubsuit$/$lower_\clubsuit$ differ per satisfaction relation and are affected by the accumulated errors. Recall the three error terms: i) discretization error~$\epsilon_\delta$, ii) the error of the approximate discretization, and iii) the PL error. We denote $Pr_{p,n}[\dMAlow,\diamond^{\leq t/\delta}G]$ with $p,n\in\{\max,\min\}$ to indicate whether player and nature are trying to maximize/minimize the reachability probability in the SG, which results from applying PL on either pMDP $\dMAlow$ or $\dMAhigh$. 

The satisfaction verification for different satisfaction relations can now be derived using the bounds below.
For feasibility and angelic-aid, due to the existential quantification over the valuations, one witness valuation $v\in \Region$ as upper bound suffices to conclude that the entire region is satisfying.
Similarly, for demonic and robustness, one counterexample suffices to conclude that the region is violating. Thus, the scheme can be optimized by terminating in such cases.
{\par\nopagebreak\small\noindent\ignorespaces%
\begin{align*}
    [lower_\clubsuit, upper_\clubsuit]\in\begin{cases}
        [Pr_{\min,\min}[\dMAlow, \diamond^{\leq t/\delta}G] , Pmin[\dMAhigh[v], \diamond^{\leq t/\delta}G]+\epsilon_\delta]\text{, if }\clubsuit=f,\\
        [Pr_{\max,\min}[\dMAlow, \diamond^{\leq t/\delta}G] , Pmax[\dMAhigh[v], \diamond^{\leq t/\delta}G]+\epsilon_\delta]\text{, if }\clubsuit=a,\\
        [Pmin[\dMAlow[v], \diamond^{\leq t/\delta}G] , Pr_{\min,\max}[\dMAhigh, \diamond^{\leq t/\delta}G]+\epsilon_\delta]\text{, if }\clubsuit=r,\\
        [Pmax[\dMAlow[v], \diamond^{\leq t/\delta}G] , Pr_{\max,\max}[\dMAhigh, \diamond^{\leq t/\delta}G]+\epsilon_\delta]\text{, if }\clubsuit=d.
    \end{cases}
\end{align*}}%

Bounds for the synthesis problem are derived similarly to the satisfaction verification.
However, the early termination criterion does not apply anymore, as we aim to classify each valuation as violating or satisfying.
{\par\nopagebreak\small\noindent\ignorespaces%
\small\begin{align*}
	[lower_\clubsuit, upper_\clubsuit]{\in}\begin{cases}
		[Pr_{\min,\min}[\dMAlow, \diamond^{\leq t/\delta}G] , Pr_{\min,\max}[\dMAhigh, \diamond^{\leq t/\delta}G]{+}\epsilon_\delta]\text{, if }\clubsuit{\in}\{f,r\},\\
		[Pr_{\max,\min}[\dMAlow, \diamond^{\leq t/\delta}G] , Pr_{\max,\max}[\dMAhigh, \diamond^{\leq t/\delta}G]{+}\epsilon_\delta]\text{, if }\clubsuit{\in}\{a,d\}.
    \end{cases}
\end{align*}}%

\noindent In practice, the synthesis is sample-guided, similar to synthesis methods for other models~\cite{DBLP:journals/fmsd/JungesAHJKQV24,DBLP:journals/acta/CeskaDPKB17}. First, a relatively fast check is performed on the instantiated MA $\MAval{v}$, for an $v\in \Region$, whose value then determines which inequality to check.

The bounds for both problems are affected by the discretization error $\epsilon_\delta$. While the parameter lifting error approaches zero as the region becomes smaller, the discretization error remains constant.
We may therefore encounter cases where the outcome remains inconclusive for every subregion. To solve problems up to arbitrary precision, one must detect this situation and decrease the discretization step $\delta$ accordingly. A conservative method is presented in Section~\ref{section:experiments}.

In summary, using the methodologies described in Sections~\ref{section:results}-\ref{section:ParSynthStrategies}, we cover the entire process as depicted in Figure~\ref{fig:process}.

\section{Experimental evaluation}\label{section:experiments}

\para{Implementation.}
We implemented our approach, supporting both pCTMCs and pMAs, using the \stormpy Python interface of the \storm probabilistic modelchecker~\cite{DBLP:journals/sttt/HenselJKQV22}.
The implementation is publicly available on GitLab.\footnote{\url{https://gitlab.tue.nl/k.v.d.glind/discretize-ma}}

The current implementation relies on explicit labels on goal-states and is able to check time-bounded reachability properties. Moreover, it can be easily extended to time-bounded until queries $\textit{safe}\; U^{\leq t}\;G$, where $\textit{safe},G\subseteq S$~\cite{DBLP:journals/corr/BraitlingFHWBH14}.

One can configure in which order to verify regions and when to decrease the discretization error, as it does not affect the theoretical guarantees. However, as it impacts performance, we briefly cover our choices. As the discretization error is the main component in performance degradation, we implemented a conservative method to decrease the discretization error. Intuitively, our approach uses that the difference between $Pr_{p,\max}[\cdot] $ and $Pr_{p,\min}[\cdot]$ tends to zero as the volume of the region decreases. Thus, if $Pr_{p,\min}[\cdot]+\epsilon_\delta>\theta $ whilst $Pr_{p,\min}[\cdot] \leq \theta$, no subregion will satisfy $Pr_{p,\max}[\cdot] +\epsilon_\delta\leq\theta$. Therefore, if $Pr_{p,\max}[\cdot] \leq \theta$ as well, splitting the region will not help and the discretization error should be reduced instead. The subregions that need to be checked after splitting are handled as a FIFO queue, ensuring that larger regions are handled first.
We implemented an additional heuristic which tracks a frontier of the already decided parameter space, allowing to prioritize regions bordering decided regions.

\begin{table}[tb]
    \caption{Benchmark statistics}
    \label{table:benchmarkStats}
    \centering
    \resizebox{0.8\textwidth}{!}{%
    \setlength{\tabcolsep}{12pt}
    \begin{tabular}{c|c|r|c|c}
        \textbf{Name} & Type &$|S|$ & $|\mathcal{X}|$ & Property \\
        \hline
erlang & pMA & 187 & 1 & $\mathbb{P}[\diamond^{\leq 5}\text{Goal}] \leq 0.2$\\
flexible-manufacturing3 & pMA & 1675 & 2 & $\mathbb{P}[\diamond^{\leq 1}\text{M3\_fail}] \leq 0.05$\\
jobs5-2 & pMA & 117 & 2 & $\mathbb{P}[\diamond^{\leq 3}\text{ `half\_of\_jobs\_finished'}] \leq 0.65$\\
polling & pMA & 990 & 2 & $\mathbb{P}[\diamond^{\leq 3}\text{q1\_full}] \leq 0.55$\\
readers-writers5 & pMA & 842 & 2 & $\mathbb{P}[\diamond^{\leq 2} \text{many\_requests}]\geq 0.01$\\
stream2-4 & pMA & 176 & 2 & $\mathbb{P}[\diamond^{\leq 3}\text{ `done'}] \leq 0.1$\\
kanban & pCTMC & 4600 & 2 & $\mathbb{P}[\diamond^{\leq 4}\text{done}]\leq0.02$\\
polling9 & pCTMC & 6912 & 2 & $\mathbb{P}[\diamond^{\leq 1}\text{first\_queue\_served}]\leq0.03$\\
tandem & pCTMC & 861 & 2 & $\mathbb{P}[\diamond^{\leq 0.3} \text{first\_queue}]\leq0.8$
    \end{tabular}}
\end{table}

\para{Set-up.}
All experiments were run on an HP Zbook G10 with a 13th Gen Intel(R) Core(TM) i7-13700H and 16GB of RAM. The implementation ran on a single core. 
Unless specified otherwise, we terminate the synthesis algorithm when at least 99.99\% of the parameter space is decided or the time limit of 15 minutes has been reached. 
The initial discretization error was set at $\theta/2$ where $\theta$ is the threshold probability as specified in Table~\ref{table:benchmarkStats}. To obtain the approximate discretized MAs $\dMAhigh, \dMAlow$, $N$ has been set to $1$ and $2$ respectively. This suffices since, as motivated in Section~\ref{ss:DependentRates}, even for small $N$ the introduced error is negligible compared to the discretization error.  
Lastly, in order to obtain an appropriate performance study of our underlying theoretical approach, all experiments were run without the frontier heuristic enabled.

\para{Benchmarks.}
We use a number of Markov Automata and CTMCs included in the QCOMP benchmark suite~\cite{DBLP:conf/tacas/HartmannsKPQR19}. Table~\ref{table:benchmarkStats} provides an overview of the models.
We focus on benchmarks where either a time-bounded property was specified or a natural extension from an unbounded reachability probability was possible.
We adapted the models to include parameters and ensured that, for every region, the original model specification is included in this region.

\para{Comparison to \prismpsy.}
\begin{table}[tb]
    \caption{\prismpsy compared to our method.}
    \label{tab:prism-psy}
    \centering\resizebox{0.8\textwidth}{!}{
    \setlength{\tabcolsep}{12pt}
    \begin{tabular}{c|c|c|c}
        Model& Region & \prismpsy & Ours \\
        \hline
         kanban & $redo3\in[0.1,0.7], ok3\in[0.2,0.9]$&99.9\%, 67s & 73.4\%, TO\\
         polling9 & $gamma\in[10,15], mu\in[0.1,0.5]$ &99.9\%, 20s & 89.1\%, TO\\
         tandem & $p\in[0.01,0.99], kappa\in[0.5,1.5]$ & not supported&84.1\%, TO 
    \end{tabular}}
\end{table}
We compare our method with \prismpsy~\cite{DBLP:journals/acta/CeskaDPKB17} on pCTMCs.
We solve the parameter synthesis problem and compare the decided fraction of the two tools within a time window of 15 minutes in Table~\ref{tab:prism-psy}.
We see that \prismpsy outperforms our approach, mainly because it does not require discretization. However, our method focuses on the novel ability to analyse pMA, which is not possible in \prismpsy. Moreover, the implementation of \prismpsy is restricted to parametric rates of the form $R(s,s') = c \cdot x_i$ for parameter $x_i$ and constant ${c>0}$. We do not impose such conditions and are able to analyze the \textit{tandem} model which involves parametric transition probabilities with fixed sojourn rates. This model cannot be analyzed with \prismpsy.

\para{Satisfaction verification of pMA.}
We check for a given benchmark model and parameter region whether the property holds under the specified satisfaction relation. 
The results are displayed in Table~\ref{tab:verificationResultsExtract}.
All benchmarks produced a conclusive answer (column S/V), and most finished within a second of total runtime, which includes both model transformation and verification. 
For more complex models, model preprocessing becomes the main factor, displaying the limited scalability of the transformations presented in Section~\ref{ss:CollapsePS}.
Table~\ref{table:regionRuntimeExtract} provides a breakdown of building and model checking times for varying region sizes and discretization errors.
Verification times correlate with the maximum sojourn rate and discretization error, following our theoretical findings. The effect of the number of parameters on verification time stems from the additional nondeterministic choices added to the parameter-free SG. For an in-depth time breakdown, we refer the reader \ifthenelse{\boolean{extended}}{%
	to Appendix~\ref{appendix:tables}.
}{%
	to the extended version~\cite{extended}.
}

\begin{table}[tb]
\ifthenelse{\boolean{extended}}{%
\caption{Satisfaction verification results. \textbf{S}atisfied/\textbf{V}iolated property. Extract of the full Table~\ref{table:verificationResults} in 
	Appendix~\ref{appendix:tables}.}}{%
    \caption{Satisfaction verification results. \textbf{S}atisfied/\textbf{V}iolated property. Extract of the full Table in~\cite{extended}.}}\label{tab:verificationResultsExtract}

\centering
\resizebox{\textwidth}{!}{%
    \setlength{\tabcolsep}{12pt}
\begin{tabular}{c|c|c|r|c}
Name & Region & Sat. relation & Time(s) & S/V \\
\hline
erlang & $R\in[7,10]$ & Feasibility & 0.13 & 1/0 \\
erlang & $R\in[7,10]$ & Angelic-aid & 0.04 & 0/1 \\
flexible-manufacturing3 & $rate1\in[0.5,2.5], rate2\in[0.5,2.5]$ & Feasibility & 0.97 & 1/0 \\
flexible-manufacturing3 & $rate1\in[0.5,2.5], rate2\in[0.5,2.5]$ & Angelic-aid & 0.97 & 1/0 \\
jobs5-2 & $x_j1\in[1,3], x_j2\in[2,4]$ & Feasibility & 0.46 & 1/0 \\
jobs5-2 & $x_j1\in[1,3], x_j2\in[2,4]$ & Angelic-aid & 0.29 & 0/1 \\
kanban & $redo3\in[0.1,0.7], ok3\in[0.2,0.9]$ & Angelic-aid & 11.85 & 1/0 \\
polling9 & $gamma\in[10,15], mu\in[0.1,0.5]$ & Angelic-aid & 5.37 & 1/0 \\
polling & $inRate1\in[2,4], inRate2\in[3,6]$ & Feasibility & 5.23 & 1/0 \\
polling & $inRate1\in[2,4], inRate2\in[3,6]$ & Angelic-aid & 6.88 & 1/0 \\
readers-writers5 & $readingRate\in[3,7], rate2\in[3,10] $& Feasibility & 0.15 & 1/0 \\
readers-writers5 & $readingRate\in[3,7], rate2\in[3,10]$ & Angelic-aid & 0.15 & 1/0 \\
stream2-4 & $inRate\in[1,7], processingRate\in[1,7]$ & Feasibility & 0.10 & 1/0 \\
stream2-4 & $inRate\in[1,7], processingRate\in[1,7]$ & Angelic-aid & 0.25 & 1/0 \\
tandem & $p\in[0.01,0.99], kappa\in[0.5,1.5]$ & Angelic-aid & 5.55 & 1/0
\end{tabular}}
\end{table}

\begin{table}[tb]
\ifthenelse{\boolean{extended}}{%
\caption{Time-breakdown of checking a single region. MC time includes all 4 possible calls $Pr_{p,n}[\dMAlow,\diamond^{\leq t/\delta}G]$. Extract of the full Table~\ref{table:regionRuntime} in Appendix~\ref{appendix:tables}.}}{%
\caption{Time-breakdown of checking a single region. MC time includes all 4 possible calls $Pr_{p,n}[\dMAlow,\diamond^{\leq t/\delta}G]$. Extract of the full Table in~\cite{extended}.}
}\label{table:regionRuntimeExtract}

\centering
\resizebox{0.8\textwidth}{!}{
\setlength{\tabcolsep}{12pt}
\begin{tabular}{c|c|c|r|r}
Benchmark & Region & $\epsilon_\delta$ & Build Time (s) & MC Time (s) \\
\hline
erlang & $R\in[1,2]$ & 0.001 & 0.023 & 0.121 \\
erlang & $R\in[1,2]$ & 0.0001 & 0.023 & 1.359 \\
erlang & $R\in[19,20]$ & 0.001 & 0.024 & 12.517 \\
erlang & $R\in[19,20]$ & 0.0001 & 0.023 & 124.381 \\
polling & $inRate1\in[2,4]$, $inRate2\in[3,6]$ & 0.001 & 5.259 & 8.109 \\
polling & $inRate1\in[2,4]$, $inRate2\in[3,6]$ & 0.0001 & 5.309 & 78.210 \\
jobs5-2-2params & $x_j1\in[1,3]$, $x_j2\in[2,4]$ & 0.001 & 0.245 & 0.562 \\
jobs5-2-3params & $x_j1\in[1,3]$, $x_j2\in[2,4]$, $x_j3\in[1,3]$ & 0.001 & 0.254 & 0.654 \\
jobs5-2-4params & $x_j1\in[1,3]$, $x_j2\in[2,4]$, $x_j3\in[1,3]$, $x_j4\in[2,4]$ & 0.001 & 0.246 & 1.166 
\end{tabular}}
\end{table}

\para{Parameter synthesis for pMA.}
\begin{table}[tb]
\caption{Parameter synthesis results. Fraction of parameter space \textbf{S}atisfying / \textbf{V}iolating / \textbf{U}nknown on property. Fractions are rounded to 3 decimals.}
\label{tab:spacePartitioning}
\centering
\resizebox{\textwidth}{!}{%
    \setlength{\tabcolsep}{12pt}
\begin{tabular}{c|c|c|r|c}
Name & Region & Sat. relation & Time(s) & S/V/U \\
\hline
erlang & $R\in[7,10]$ & Demonic & 0.03 & 0.000/1.000/0.000 \\
erlang & $R\in[7,10]$ & Robust & 336.38 & 0.311/0.689/$<10^{-4}$ \\
flexible-manufacturing3 & $rate1\in[0.5,2.5], rate2\in[0.5,2.5]$ & Demonic & TO & 0.687/0.291/0.022 \\
flexible-manufacturing3 & $rate1\in[0.5,2.5], rate2\in[0.5,2.5]$ & Robust & TO & 0.687/0.291/0.022 \\
jobs5-2 & $x_j1\in[1,3], x_j2\in[2,4]$ & Demonic & 0.31 & 0.000/1.000/0.000 \\
jobs5-2 & $x_j1\in[1,3], x_j2\in[2,4]$ & Robust & TO & 0.416/0.577/0.007 \\
kanban & $redo3\in[0.1,0.7], ok3\in[0.2,0.9]$ & Robust & TO & 0.570/0.164/0.266 \\
polling9 & $gamma\in[10,15], mu\in[0.1,0.5]$ & Robust & TO & 0.884/0.007/0.109 \\
polling & $inRate1\in[2,4], inRate2\in[3,6]$ & Demonic & TO & 0.199/0.784/0.017 \\
polling & $inRate1\in[2,4], inRate2\in[3,6]$ & Robust & TO & 0.900/0.083/0.017 \\
readers-writers5 & $readingRate\in[3,7], rate2\in[3,10]$ & Demonic & TO & 0.617/0.031/0.352 \\
readers-writers5 & $readingRate\in[3,7], rate2\in[3,10]$ & Robust & TO & 0.617/0.031/0.352 \\
stream2-4 & $inRate\in[1,7], processingRate\in[1,7]$ & Demonic & TO & 0.494/0.496/0.010 \\
stream2-4 & $inRate\in[1,7], processingRate\in[1,7]$ & Robust & TO & 0.830/0.162/0.007 \\
tandem & $p\in[0.01,0.99], kappa\in[0.5,1.5]$ & Robust & TO & 0.500/0.314/0.186
\end{tabular}}
\end{table}
The synthesis results are displayed in Table~\ref{tab:spacePartitioning}. Note that the checked regions are exactly the same as in Table~\ref{tab:verificationResultsExtract}. For each benchmark, we checked robust and demonic satisfaction relations. In accordance with Figure~\ref{fig:strategyRelations}, the fraction of satisfying valuations for the robust satisfaction relation is always at least as large as that for the demonic relation.

\begin{figure}[tb]
\hspace{-1cm}
\centering
\begin{subfigure}[t]{0.35\textwidth}
\includegraphics[scale=0.35]{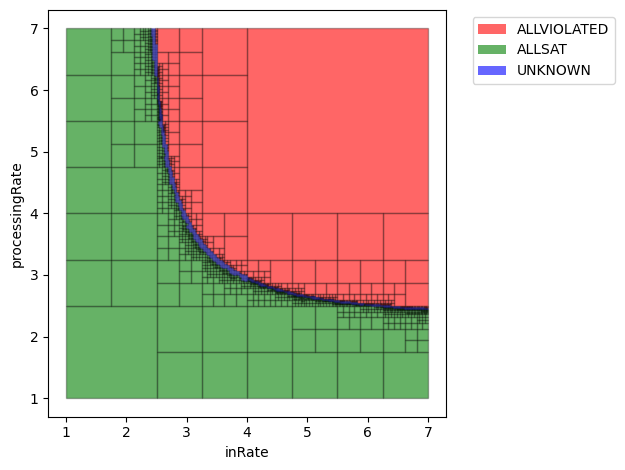}
\caption{Demonic}
\end{subfigure}%
~\hspace{1cm}
\begin{subfigure}[t]{0.35\textwidth}
\includegraphics[scale=0.35]{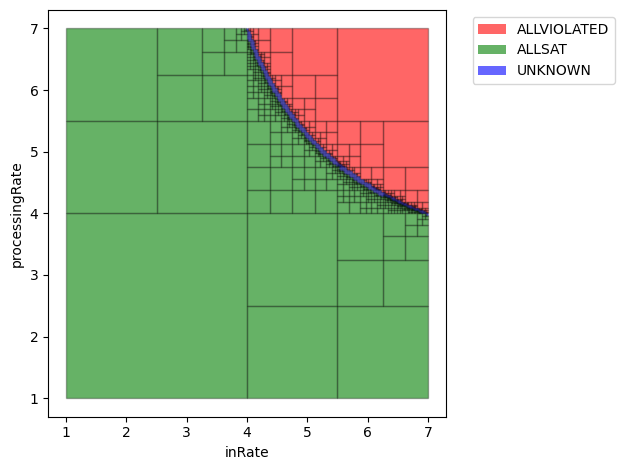}
\caption{Robust}
\end{subfigure}
    \caption{Parameter synthesis on \textit{stream2-4} under different satisfaction relations.}
    \label{fig:streamingSpacePartitioning}
\end{figure}
Figure~\ref{fig:streamingSpacePartitioning} shows the space partitioning obtained for the \textit{stream} benchmark. Each rectangle in the space gives the region on which this decision was based. 
We note that extensive portions of the parameter space can be resolved quickly, while more detailed analysis is necessary near the boundary between satisfaction and violation. \ifthenelse{\boolean{extended}}{%
	Appendix~\ref{appendix:spacePartitionFigures} provides the synthesis figures for all other benchmarks.
}{%
	Synthesis figures for all other benchmarks are provided in~\cite{extended}.
} 

\begin{figure}[tb]
    \centering
    \resizebox{0.7\textwidth}{!}{
    \includegraphics[]{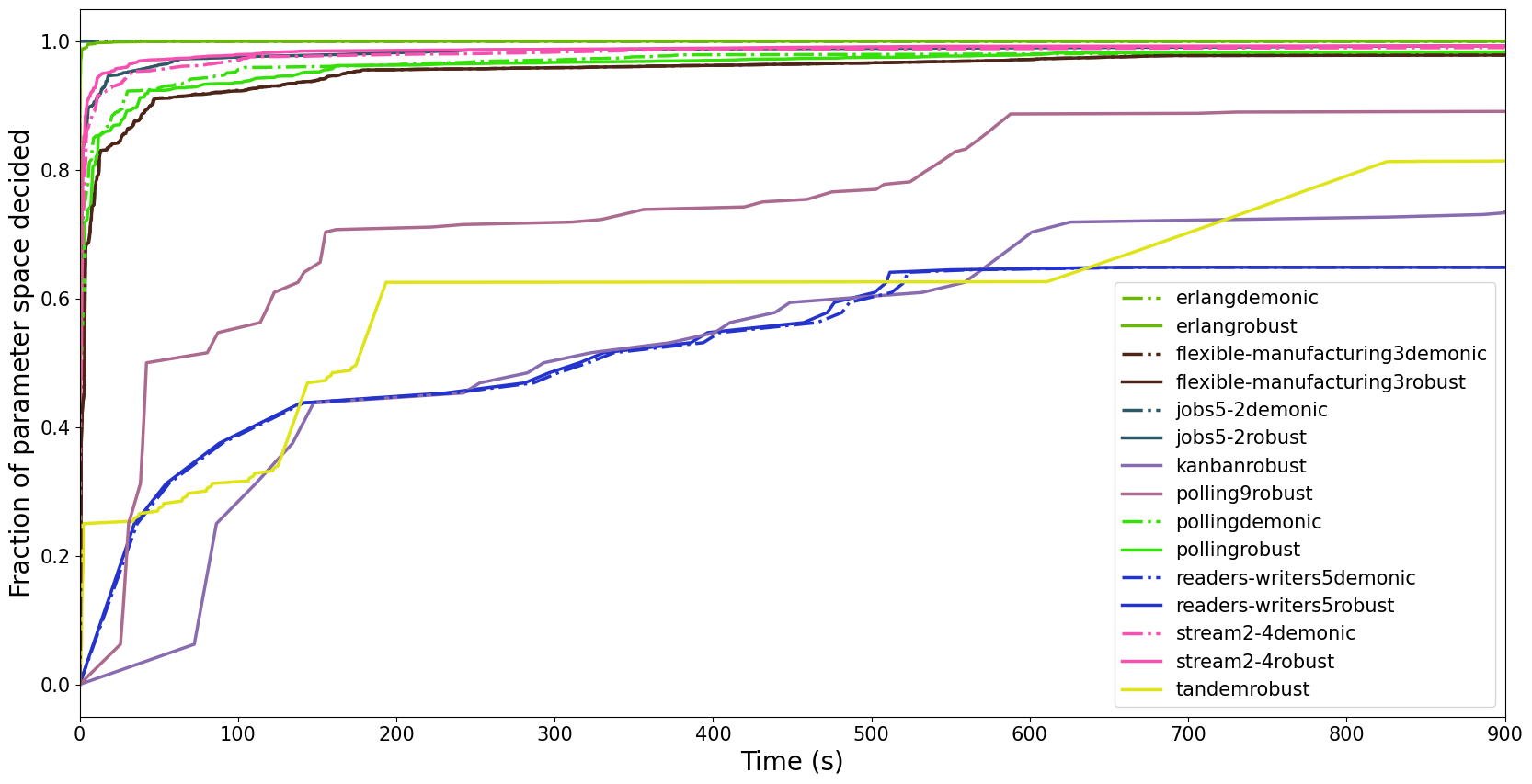}}
    \caption{Ratio of decided area over runtime for pMA benchmarks.}
    \label{fig:RuntimeVsAreaDecided}
\end{figure}
In Table~\ref{tab:spacePartitioning}, we observe that almost all benchmarks reach the timeout. This occurs for multiple reasons.
First, determining whether a region satisfies a property requires a fixed number of model checking queries. As each subregion is progressively refined, accounting for an increasingly smaller portion of the overall parameter space, an equivalent computational effort is required for a diminished gain. This phenomenon is also observed in other methods, including work on discrete-time models~\cite{DBLP:journals/fmsd/JungesAHJKQV24,DBLP:conf/rtss/HanKM08,DBLP:journals/acta/CeskaDPKB17}. Figure~\ref{fig:RuntimeVsAreaDecided} illustrates this perceived stagnation, where the decided fraction over time is plotted. In theory one converges to the total area decided, in practice we observe an exponential deceleration of gain in decided parameter space.
Second, as it is necessary to decrease the discretization error, the computational cost of each model-checking query is increased as well. 
Therefore, the perceived deceleration in Figure~\ref{fig:RuntimeVsAreaDecided} is exacerbated.

\section{Conclusion}
\label{section:conclusion}

We introduced the notion of parametric Markov Automata, extending Markov Automata with parametric rate- and transition functions. We presented methods to perform satisfaction verification and parameter synthesis on these models. Our approach consists of two steps. First, we transform the pMA to an approximate discretized pMA. Second, we analyze the approximate behaviour of the pMA by applying existing parameter lifting techniques.
The experimental evaluation indicated that both the predominant source of error and computational demand arise from the discretization method.

\para{Future work.}
Given our findings above, future research could investigate approaches which do not use the discretization method. This could for instance require the analysis of continuous-time stochastic games without discretization in order to perform parameter lifting. Alternatively, one could opt for an approach using uniformization, such as in~\cite{DBLP:conf/atva/ButkovaHHK15}, instead. However, this would most likely require many costly calls on the SG due to the prophetic scheduler.

A better choice of heuristics may optimize the total decided area for the parameter synthesis problem. That is, one may investigate how to use previous information to decide when to reduce the discretization error. Less conservative conditions on when to decrease the discretization error could potentially save expensive model checking queries on the discretized MA. However, it may lead to, potentially, checking regions with a lower discretization error than necessary. 

We believe our method can also be extended to time-interval bounded queries. Here, the main difficulty lies in finding an appropriate approximate discretization such that lower and upper bounds are guaranteed.

\para{Data availability.}
The implementation and benchmark files are publicly available in the artifact at \href{https://doi.org/10.5281/zenodo.19661781}{doi.org/10.5281/zenodo.19661781}.

\bibliographystyle{splncs04}
\bibliography{refs}

\ifextended
    \clearpage
    \appendix
    \section{Parameter synthesis figures}\label{appendix:spacePartitionFigures}
We provide the parameter synthesis results for all considered models.
Visualizations for the pMA models under both the demonic and robust satisfaction relation are given in Figures~\ref{fig:synStreaming}, \ref{fig:synReaderWriter}, \ref{fig:synJobs}, \ref{fig:synManufacturing}, \ref{fig:synPolling}.
Note that the visualization for \textit{stream2-4} in Figure~\ref{fig:synStreaming} was also already given in the main text in Figure~\ref{fig:streamingSpacePartitioning}.
The \textit{erlang} model is not depicted as it only has one parameter.

For pCTMCs, no non-determinism exists and the robust and demonic satisfaction relations coincide.
The visualizations for the pCTMCs are given in Figures~\ref{fig:synPollingCTMC}, \ref{fig:synTandem}, \ref{fig:synKanban}.

\begin{figure}[!htb]
\centering
\begin{subfigure}[t]{0.4\textwidth}
\includegraphics[scale=0.4]{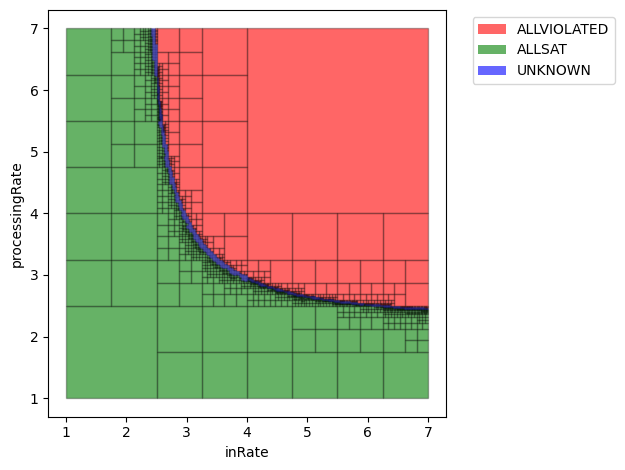}
\caption{Demonic}
\end{subfigure}%
~\hspace{2cm}
\begin{subfigure}[t]{0.4\textwidth}
\includegraphics[scale=0.4]{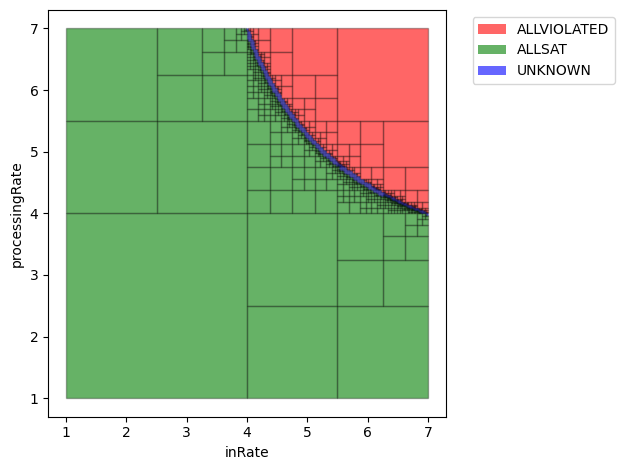}
\caption{Robust}
\end{subfigure}
    \caption{Parameter synthesis on \textit{stream2-4} under different satisfaction relations.}
    \label{fig:synStreaming}
\end{figure}

\begin{figure}[!htb]
\centering
\begin{subfigure}[t]{0.4\textwidth}
\includegraphics[scale=0.4]{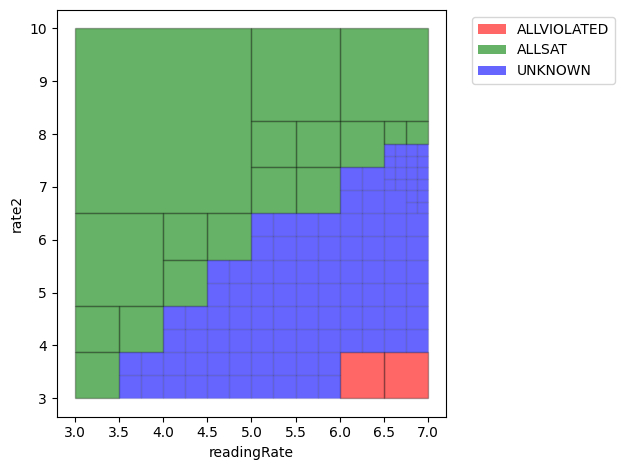}
\caption{Demonic}
\end{subfigure}%
~\hspace{2cm}
\begin{subfigure}[t]{0.4\textwidth}
\includegraphics[scale=0.4]{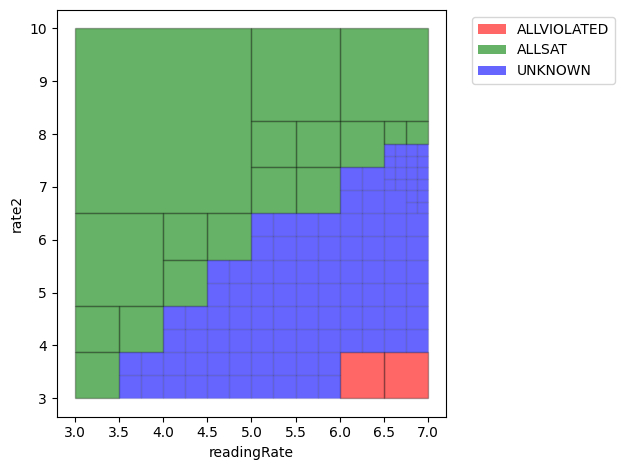}
\caption{Robust}
\end{subfigure}
    \caption{Parameter synt. on \textit{readers-writers5} under different satisfaction relations.}
    \label{fig:synReaderWriter}
\end{figure}

\begin{figure}[!htb]
\centering
\begin{subfigure}[t]{0.4\textwidth}
\includegraphics[scale=0.4]{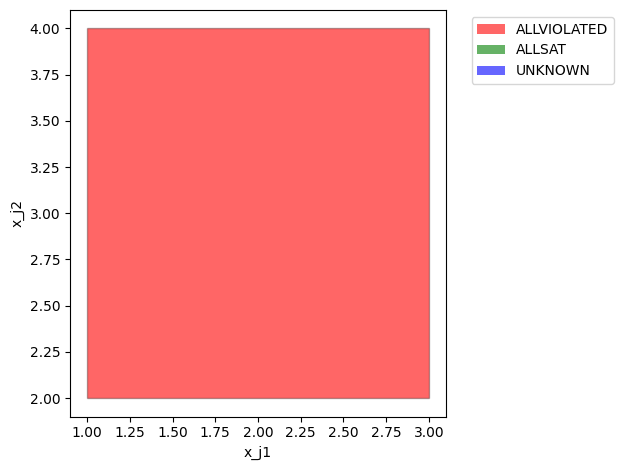}
\caption{Demonic}
\end{subfigure}%
~\hspace{2cm}
\begin{subfigure}[t]{0.4\textwidth}
\includegraphics[scale=0.4]{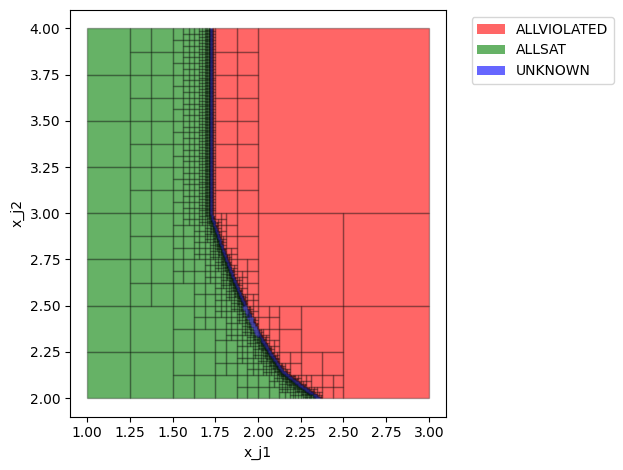}
\caption{Robust}
\end{subfigure}
    \caption{Parameter synthesis on \textit{jobs5-2} under different satisfaction relations.}
    \label{fig:synJobs}
\end{figure}

\begin{figure}[!htb]
\centering
\begin{subfigure}[t]{0.4\textwidth}
\includegraphics[scale=0.4]{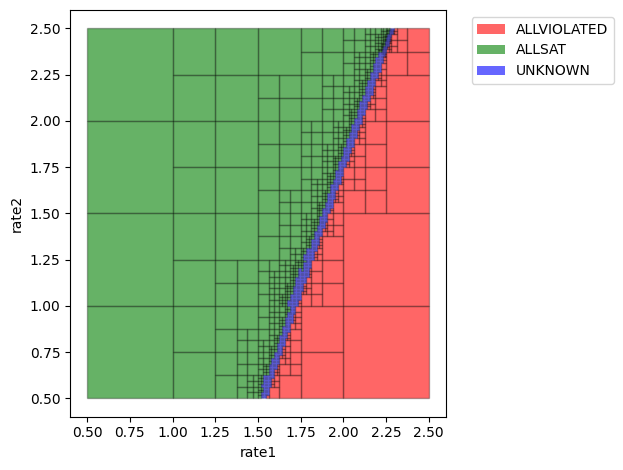}
\caption{Demonic}
\end{subfigure}%
~\hspace{2cm}
\begin{subfigure}[t]{0.4\textwidth}
\includegraphics[scale=0.4]{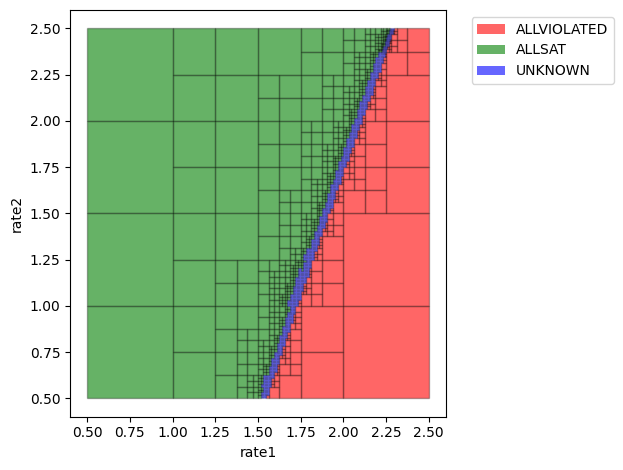}
\caption{Robust}
\end{subfigure}
    \caption{Paramet. synt. on \textit{flexible-manufacturing3} under diff. satisf. relations.}
    \label{fig:synManufacturing}
\end{figure}

\begin{figure}[!htb]
\centering
\begin{subfigure}[t]{0.4\textwidth}
\includegraphics[scale=0.4]{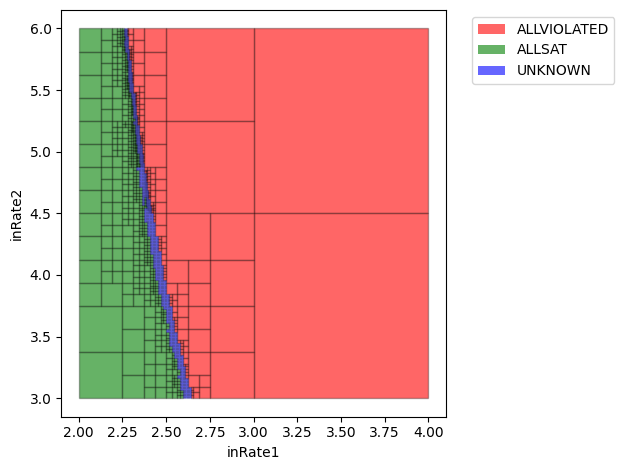}
\caption{Demonic}
\end{subfigure}%
~\hspace{2cm}
\begin{subfigure}[t]{0.4\textwidth}
\includegraphics[scale=0.4]{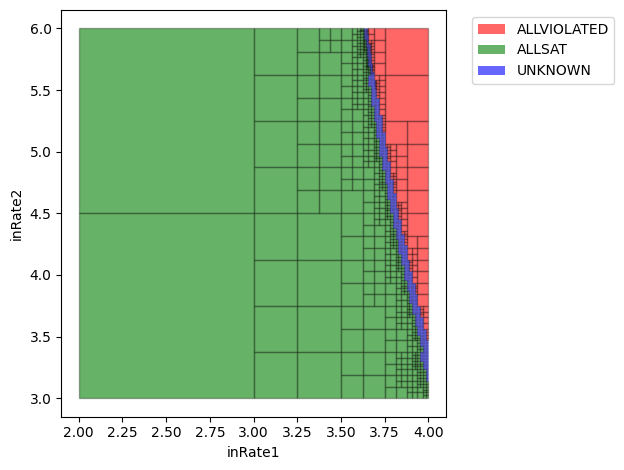}
\caption{Robust}
\end{subfigure}
    \caption{Parameter synthesis on \textit{polling} under different satisfaction relations.}
    \label{fig:synPolling}
\end{figure}

\begin{figure}[!htb]
\centering
\includegraphics[scale=0.4]{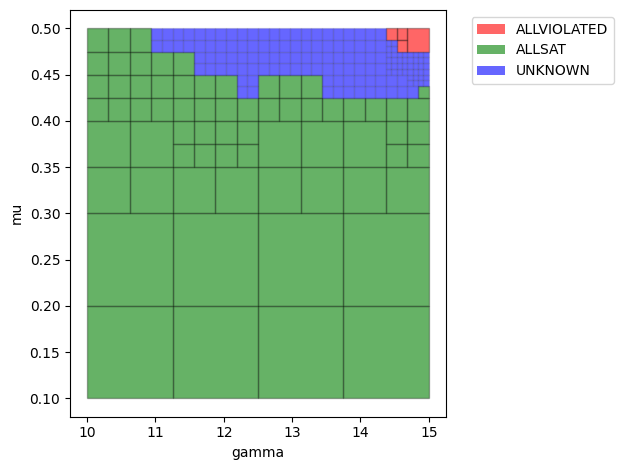}
    \caption{Parameter synthesis on \textit{polling9} pCTMC.}
    \label{fig:synPollingCTMC}
\end{figure}

\begin{figure}
\begin{minipage}[c]{0.4\textwidth}
\includegraphics[scale=0.4]{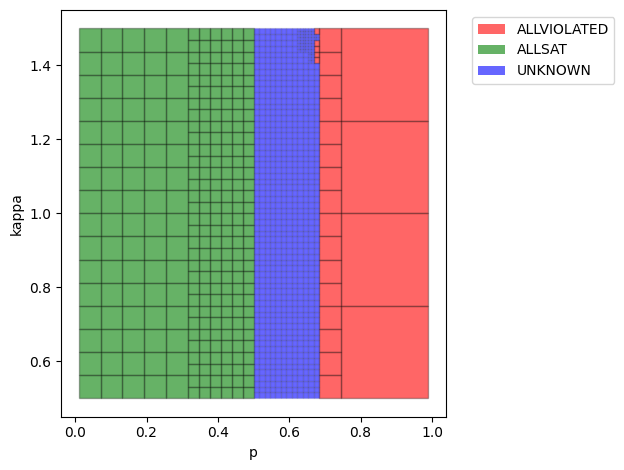}
    \caption{Parameter synthesis on \textit{tandem} pCTMC.}
    \label{fig:synTandem}
\end{minipage}
\hfill
\begin{minipage}[c]{0.4\textwidth}
\includegraphics[scale=0.4]{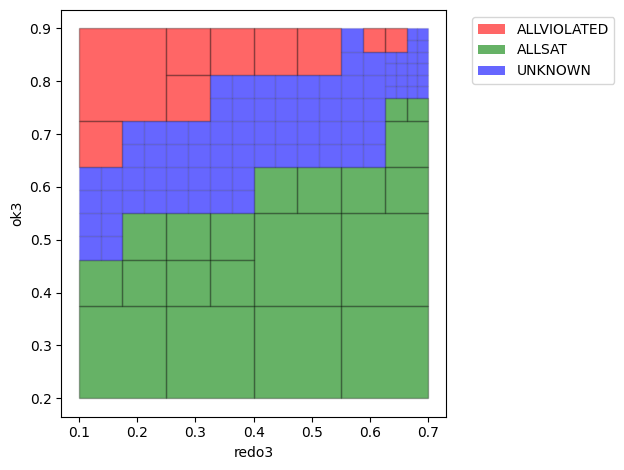}
    \caption{Parameter synthesis on \textit{kanban} pCTMC.}
    \label{fig:synKanban}
\end{minipage}%
\end{figure}

\FloatBarrier

\section{Tables}\label{appendix:tables}

This section presents all measurements obtained for the satisfaction verification  and the time breakdown.

Table~\ref{table:verificationResults} provides the satisfaction verification results for all considered models. It is the full version of Table~\ref{tab:verificationResultsExtract} from the main text.

Table~\ref{table:regionRuntime} gives a time-breakdown of checking a single region.
The table is the full version of Table~\ref{table:regionRuntimeExtract} from the main text.
The modelchecking times in the table are a sum of all possible combinations of: $Pr_{p,n}[\dMAlow,\diamond^{\leq t/\delta}G]$ with ${p,n\in\{\max,\min\}}$. Using this table, we can confirm some observations made in previous sections. The building time on pMA benchmarks increases significantly when we enlarge the state space. Moreover, with respect to modelchecking runtimes, we observe a close to inverse linear relationship with the discretization error $\epsilon_\delta$ and a positive correlation with the maximum value in the region.
We observe that the building time is independent of these factors.

This correlation between maximum value and the model checking times can be attributed to the fact that, in all benchmarks, a larger value in the region results in a higher maximum sojourn rate. Consequently, a smaller discretization step $\delta$ is required to ensure the same error bound $\epsilon_\delta$. Naturally, decreasing $\epsilon_\delta$ requires a smaller discretization step as well, thus increasing modelchecking times. That is, these relations between discretization- step and error can be derived from Lemma~\ref{lemma:parametricMADiscretization} and the increase of model checking times is immediately derived from the fact that a smaller discretization step yields a model checking query on the stochastic game with a larger stepbound. Thereby, this call is inherently more expensive. 

When increasing the number of parameters in the model, as showcased for the jobs benchmark, we see a slight increase in model checking times as well which can be attributed to the introduction of additional nondeterministic choices in the parameter-free stochastic game. More significantly, increasing the number of parameters increases the dimension of the region $\Region$, such that every indecisive region will be divided in $2^{|\mathcal{X}|}$ subregions. That is, the number of additional subregions introduced compared to the original model is exponential in terms of the number of auxiliary variables.

\begin{table}[!htb]
\caption{Satisfaction verification results. \textbf{S}atisfied/\textbf{V}iolated property.}
\label{table:verificationResults}
\centering
\resizebox{\textwidth}{!}{%
    \setlength{\tabcolsep}{12pt}
\begin{tabular}{c|c|c|r|c}
Name & Region & Sat. relation & Time(s) & S/V \\
\hline

erlang & $R\in[7,10]$ & Demonic & 0.03 & 0/1 \\
erlang & $R\in[7,10]$ & Feasibility & 0.13 & 1/0 \\
erlang & $R\in[7,10]$ & Robust & 0.02 & 0/1 \\
erlang & $R\in[7,10]$ & Angelic-aid & 0.04 & 0/1 \\
flexible-manufacturing3 & $rate1\in[0.5,2.5], rate2\in[0.5,2.5]$ & Demonic & 0.47 & 0/1 \\
flexible-manufacturing3 & $rate1\in[0.5,2.5], rate2\in[0.5,2.5]$ & Feasibility & 0.97 & 1/0 \\
flexible-manufacturing3 & $rate1\in[0.5,2.5], rate2\in[0.5,2.5]$ & Robust & 0.44 & 0/1 \\
flexible-manufacturing3 & $rate1\in[0.5,2.5], rate2\in[0.5,2.5]$ & Angelic-aid & 0.97 & 1/0 \\
jobs5-2 & $x_j1\in[1,3], x_j2\in[2,4]$ & Demonic & 0.31 & 0/1 \\
jobs5-2 & $x_j1\in[1,3], x_j2\in[2,4]$ & Feasibility & 0.46 & 1/0 \\
jobs5-2 & $x_j1\in[1,3], x_j2\in[2,4]$ & Robust & 0.31 & 0/1 \\
jobs5-2 & $x_j1\in[1,3], x_j2\in[2,4]$ & Angelic-aid & 0.29 & 0/1 \\
kanban & $redo3\in[0.1,0.7], ok3\in[0.2,0.9]$ & Robust & 17.23 & 0/1 \\
kanban & $redo3\in[0.1,0.7], ok3\in[0.2,0.9]$ & Angelic-aid & 11.85 & 1/0 \\
polling9 & $gamma\in[10,15], mu\in[0.1,0.5]$ & Robust & 50.06 & 0/1 \\
polling9 & $gamma\in[10,15], mu\in[0.1,0.5]$ & Angelic-aid & 5.37 & 1/0 \\
polling & $inRate1\in[2,4], inRate2\in[3,6]$ & Demonic & 5.25 & 0/1 \\
polling & $inRate1\in[2,4], inRate2\in[3,6]$ & Feasibility & 5.23 & 1/0 \\
polling & $inRate1\in[2,4], inRate2\in[3,6]$ & Robust & 5.42 & 0/1 \\
polling & $inRate1\in[2,4], inRate2\in[3,6]$ & Angelic-aid & 6.88 & 1/0 \\
readers-writers5 & $readingRate\in[3,7], rate2\in[3,10]$ & Demonic & 36.00 & 0/1 \\
readers-writers5 & $readingRate\in[3,7], rate2\in[3,10] $& Feasibility & 0.15 & 1/0 \\
readers-writers5 & $readingRate\in[3,7], rate2\in[3,10]$ & Robust & 186.91 & 0/1 \\
readers-writers5 & $readingRate\in[3,7], rate2\in[3,10]$ & Angelic-aid & 0.15 & 1/0 \\
stream2-4 & $inRate\in[1,7], processingRate\in[1,7]$ & Demonic & 0.10 & 0/1 \\
stream2-4 & $inRate\in[1,7], processingRate\in[1,7]$ & Feasibility & 0.10 & 1/0 \\
stream2-4 & $inRate\in[1,7], processingRate\in[1,7]$ & Robust & 0.15 & 0/1 \\
stream2-4 & $inRate\in[1,7], processingRate\in[1,7]$ & Angelic-aid & 0.25 & 1/0 \\
tandem & $p\in[0.01,0.99], kappa\in[0.5,1.5]$ & Robust & 1.42 & 0/1 \\
tandem & $p\in[0.01,0.99], kappa\in[0.5,1.5]$ & Angelic-aid & 5.55 & 1/0 \\
\end{tabular}}
\end{table}

\begin{table}[!htb]
\caption{Timings of different benchmarks, varying in region and discretization error.}
\label{table:regionRuntime}
\centering
\resizebox{\textwidth}{!}{
\begin{tabular}{c|c|l|r|r}
Benchmark & Region & $\epsilon_\delta$ & Build Time (s) & MC Time (s) \\
\hline
erlang & $R\in[1,10]$ & 0.001 & 0.024 & 3.388 \\
erlang & $R\in[1,10]$ & 0.0001 & 0.024 & 32.274 \\
erlang & $R\in[1,4]$ & 0.001 & 0.024 & 0.534 \\
erlang & $R\in[1,4]$ & 0.0001 & 0.023 & 5.190 \\
erlang & $R\in[7,10]$ & 0.001 & 0.024 & 2.932 \\
erlang & $R\in[4,5]$ & 0.001 & 0.023 & 0.819 \\
erlang & $R\in[4,5]$ & 0.0001 & 0.022 & 8.117 \\
erlang & $R\in[7,10]$ & 0.0001 & 0.024 & 31.025 \\
erlang & $R\in[4,8]$ & 0.001 & 0.024 & 2.123 \\
erlang & $R\in[4,8]$ & 0.0001 & 0.024 & 20.589 \\
erlang & $R\in[1,2]$ & 0.001 & 0.023 & 0.121 \\
erlang & $R\in[1,2]$ & 0.0001 & 0.023 & 1.359 \\
erlang & $R\in[19,20]$ & 0.001 & 0.024 & 12.517 \\
erlang & $R\in[19,20]$ & 0.0001 & 0.023 & 124.381 \\
erlang & $R\in[1,20]$ & 0.001 & 0.023 & 12.370 \\
erlang & $R\in[1,20]$ & 0.0001 & 0.015 & 122.813 \\
flexible-manufacturing3 & $rate1\in[0.5,2.5]$, $rate2\in[0.5,2.5]$ & 0.001 & 0.246 & 1.303 \\
flexible-manufacturing3 & $rate1\in[0.5,2.5]$, $rate2\in[0.5,2.5]$ & 0.0001 & 0.245 & 11.357 \\
jobs5-2-2params & $x_j1\in[1,3]$, $x_j2\in[2,4]$ & 0.001 & 0.245 & 0.562 \\
jobs5-2-2params & $x_j1\in[1,3]$, $x_j2\in[2,4]$ & 0.0001 & 0.303 & 5.374 \\
jobs5-2-3params & $x_j1\in[1,3]$, $x_j2\in[2,4]$, $x_j3\in[1,3]$ & 0.001 & 0.254 & 0.654 \\
jobs5-2-3params & $x_j1\in[1,3]$, $x_j2\in[2,4]$, $x_j3\in[1,3]$ & 0.0001 & 0.297 & 6.427 \\
jobs5-2-4params & $x_j1\in[1,3]$, $x_j2\in[2,4]$, $x_j3\in[1,3]$, $x_j4\in[2,4]$ & 0.001 & 0.246 & 1.166 \\
jobs5-2-4params & $x_j1\in[1,3]$, $x_j2\in[2,4]$, $x_j3\in[1,3]$, $x_j4\in[2,4]$ & 0.0001 & 0.273 & 10.098 \\
jobs5-2-5params & $x_j1\in[1,3]$, $x_j2\in[2,4]$, $x_j3\in[1,3]$, $x_j4\in[2,4]$, $x_j5\in[2,4]$ & 0.001 & 0.246 & 1.335 \\
jobs5-2-5params & $x_j1\in[1,3]$, $x_j2\in[2,4]$, $x_j3\in[1,3]$, $x_j4\in[2,4]$, $x_j5\in[2,4]$ & 0.0001 & 0.245 & 12.518 \\
kanban & $redo3\in[0.1,0.7]$, $ok3\in[0.2,0.9]$ & 0.001 & 1.007 & 52.770 \\
kanban & $redo3\in[0.1,0.7]$, $ok3\in[0.2,0.9]$ & 0.0001 & 0.962 & 524.605 \\
polling9 & $gamma\in[10,15]$, $mu\in[0.1,0.5]$ & 0.001 & 1.372 & 26.216 \\
polling9 & $gamma\in[10,15]$, $mu\in[0.1,0.5]$ & 0.0001 & 1.300 & 253.763 \\
polling & $inRate1\in[2,4]$, $inRate2\in[3,6]$ & 0.001 & 5.259 & 8.109 \\
polling & $inRate1\in[2,4]$, $inRate2\in[3,6]$ & 0.0001 & 5.309 & 78.210 \\
readers-writers5 & $readingRate\in[3,7]$, $rate2\in[3,10]$ & 0.001 & 0.086 & 72.818 \\
readers-writers5 & $readingRate\in[3,7]$, $rate2\in[3,10]$ & 0.0001 & 0.087 & 735.668 \\
stream2-4 & $inRate\in[1,7]$, $processingRate\in[1,7]$ & 0.001 & 0.089 & 8.111 \\
stream2-4 & $inRate\in[1,7]$, $processingRate\in[1,7]$ & 0.0001 & 0.094 & 79.102 \\
tandem & $p\in[0.01,0.99]$, $kappa\in[0.5,1.5]$ & 0.001 & 0.103 & 6.817 \\
tandem & $p\in[0.01,0.99]$, $kappa\in[0.5,1.5]$ & 0.0001 & 0.102 & 69.817 \\
\end{tabular}}
\end{table}

\FloatBarrier
\section{Probabilistic State collapse algorithm}\label{appendix:algorithm}
After discretizing the pMA, the remaining $\PS$ states must be collapsed into Markovian states in order to compute the step-bounded reachability probability.
For this collapse, we must be able to distinguish which states originally belonged to $\MS$ and $\PS$. This is done by adding a reward structure $R$ to the discretization such that:
\begin{align*}
R(s) = \begin{cases}
    1, \textit{ if }s\in MS\\
    0, \textit{ if }s\in PS
\end{cases}
\end{align*}
\begin{definition}[Binary reward MDP]\label{def:binaryMDP}
    We call an MDP $M = \{S,A,s_0, \rightarrow, R\}$ with reward structure $R$ a \emph{binary reward MDP} if $R(s) \in \{0,1\}$ for all $s \in S$.
    Moreover, we partition $S=S_1\sqcup S_0$ into sets of states with reward $1$ and $0$, respectively.
\end{definition}
It is easy to see that cost-bounded reachability with respect to $R$ in this adapted discretized pMA yields the same result as the adapted VI-scheme. Algorithm~\ref{alg:BinaryRewardMDP} ensures that, given a binary reward MDP without $0$-reward loops, it returns an MDP for which the step-bounded reachability is equivalent to the cost-bounded reachability in the former. In particular, for our use case, it means that Algorithm~\ref{alg:BinaryRewardMDP} collapses the remaining zero-reward states of the discretized model into states with reward one.
The algorithm follows state-elimination ideas such as in~\cite{10.1007/978-3-319-47677-3_6}.
Note that our discretized pMA is guaranteed to be zero-loop free due to preprocessing as covered in Section~\ref{ss:CollapsePS}. Therefore, the following result is obtained:

\begin{lemma}
    Given a discretized pMA $\dMA$ with the reward structure defined above, $G\subseteq \MS$ and $N\in \mathbb{N}$. Let $\dMA_1$ be the model returned from Algorithm~\ref{alg:BinaryRewardMDP}. Then we have the following:
    \begin{align*}
        Pmax[\dMA, \diamond^{c_R\leq N} G] = Pmax[\dMA_1, \diamond^{\leq N} G] 
    \end{align*}
    Moreover, the same holds for $Pmin[\cdot]$.
\end{lemma}
\begin{proof}[Sketch]
    For any valuation $v\in R$, the equality holds due to the reasoning above. The inequality then holds for the parametric models as well. \qed
\end{proof}

\begin{algorithm}[tb]
\caption{Converting binary reward MDP to MDP}\label{alg:BinaryRewardMDP}
\begin{algorithmic}
\Require Binary reward MDP $M=(S_0\sqcup S_1, s_0, Act, \rightarrow, R)$ with $s_0\in S_1$ and $M$ free of zero-reward loops.
\Ensure $Pr^{M}[\diamond^{c\in I}G] = Pr^{M'}[\diamond^{\in I}G]$ for any $G\subseteq S_1, I=[t_0,t_1]\subset \mathbb{R}_{\geq 0}$, \\\hspace{1.2cm}$\forall (s,s')\in S_1\times S_0: Paths^{M'}_{s}(s') = \emptyset$

\State $M' \gets \emptyset$
\Repeat
\If{$M' = \emptyset$}
\State $M'' \gets M$
\Else
\State $M'' = M'$
\EndIf
\For {$s_1\in S, a\in Act(s_1)$}
\State \Comment Gather all $\PS$ successor states of $s_1$ and remove transitions to them.
\State $ReachableZeroRewardStates \gets \{s\in S_0~|~s\in supp(P(s_1,a,\cdot)) \}$
\For {$s\in ReachableZeroRewardStates$}
\State $P'(s_1,a,s) \gets 0$
\EndFor
\State \Comment Cartesian product of outgoing actions for all $\PS$
	\State $ActionCombs \gets \times_{s'_i\in~ReachableZeroRewardStates} \{(s'_i, a'_i) |~a'_i\in~Acts(s'_i)\}$
\State \Comment Set transitions from $s_i$ to successor states of $\PS$.
\For {$Acts \in ActionCombs, t \in S$}
	\State $P'(s_1, (s_1, a) \times Acts, t) \gets P(s_1,a,t)~+~\sum_{(s'_i, a'_i) \in Acts}P(s_1,a,s'_i)P(s'_i,a'_i,t)$
\EndFor
\If {$ReachableZeroRewardStates = \emptyset$}
\State  $P'(s_1, (s_1, a), t) \gets P(s_1,a,t)$
\EndIf
\EndFor
\State \Comment Update model.
\State RemoveAnyUnreachableStates()
\State $M' = (S_1\sqcup S_0, s_0, Act, P')$
\Until{$M''\neq M'$}\\
\Return $M'$
\end{algorithmic}
\end{algorithm}

\FloatBarrier

\section{Approximate discretization}\label{appendix:approximateDiscretization}
In this section, we will prove Theorem~\ref{thm:approxDiscretizationError} in two steps. Lemma~\ref{lemma:pMAapproxApproximations} guarantees lower and upper bounds with respect to the true reachability probability and Lemma~\ref{lemma:approxDiscretizationBounds} guarantees that this additional error tends to zero as the discretization step goes to zero.
\begin{lemma}\label{lemma:pMAapproxApproximations}
	Given pMA $\MA{=}(S,A, s_0,\rightarrow, \Rightarrow)$ and discretization step $\delta{>}0$, we~have:
	\[
		Pmax[\dMAlow,\diamond^{\leq t/\delta}G] \leq Pmax[\MA,\diamond^{\leq t}G] \leq Pmax[\dMAhigh,\diamond^{\leq t/\delta}G] + \epsilon_\delta,\] where $\epsilon_\delta$ is the discretization error as per Lemma~\ref{lemma:parametricMADiscretization}. 
    Moreover, the same result holds for $Pmin[\cdot]$. 
\end{lemma}
\begin{proof}
    For $\dMAlow, \dMAhigh$, the exponential function is over- and under-approximated, as per Equation~\ref{eq:expApprox}. This exponential is the probability of a self-loop in the approximate discretization of $\MA$. Therefore, the probability mass of taking a self loop is, respectively, over- and under-approximated. As higher self-loop probabilities lead to smaller probabilities for time-bounded reachability, we get the following chain of inequalities.
    
    \begin{align*}
        &Pmax[\dMAlow, \diamond^{\leq t/\delta}G] \leq Pmax[\dMA, \diamond^{\leq t/\delta}G] \\
        &\leq_{\text{Lemma}~\ref{lemma:parametricMADiscretization}}Pmax[\MA, \diamond^{\leq t}G]\\
        &\leq_{\text{Lemma}~\ref{lemma:parametricMADiscretization}}Pmax[\dMA,\diamond^{\leq t/\delta}G] +\epsilon_\delta\leq Pmax[\dMAhigh, \diamond^{\leq t/\delta}G] + \epsilon_\delta
    \end{align*}
    The case $Pmin[\cdot]$ is analogous; thus we conclude the proof.\qed
\end{proof}

\begin{lemma}\label{lemma:approxDiscretizationBounds}
    Given pMA $\MA=(S,A,s_0,\rightarrow, \Rightarrow)$, discretization step $\delta>0$ and $N>0$. Let $\dMAlow = \dMA_{2N} , \dMAhigh = \dMA_{2N+1}$. We have that:
    \begin{align*}
	    Pmax[\dMA,\diamond^{\leq t/\delta}G] - \epsilon^{2N}_{t/\delta} &\leq Pmax[\dMAlow,\diamond^{\leq t/\delta}G]\\
	    Pmax[\dMAhigh,\diamond^{\leq t/\delta}G] &\leq Pmax[\dMA,\diamond^{\leq t/\delta}G] + \epsilon^{2N+1}_{t/\delta},
    \end{align*}
    where $\epsilon^M_{t/\delta} = (1+\frac{(\lambda\delta)^M}{M!})^{t/\delta}-1$ with $\lambda =\sup_{v\in\Region}\lambda[v]$.
    Moreover, the same result holds for $Pmin(\cdot)$. 
\end{lemma}
\begin{proof}
    Let $Sched(M)$ denote the set of all positional schedulers on model $M$. We first show that: $Pmax[\dMA_s,\diamond^{\leq k}G] - \epsilon^{2N}_{k} \leq Pmax[\dMAlow_s,\diamond^{\leq k}G]$ for all $k\in\mathbb{N}$ and $s \in S$. Here one should not forget that we should only count steps if taken in Markovian states.

	For case $k=0$, equality holds, because no steps are taken. Therefore, we continue our proof with induction, using base case $k=1$: 
       
\begin{itemize}
\item Suppose that $s\in \MS$, then:
        \begin{align*}
            &Pmax[\dMAlow_s,\diamond^{\leq 1}G] = (1-\sum_{j=0}^{2N}\frac{(-\lambda(s)\delta)^j}{j!})\sum_{s'\in S}p(s,s')Pmax[\dMAlow_{s'},\diamond^{\leq 0}G]\\
            =& (1-\sum_{j=0}^{2N}\frac{(-\lambda(s)\delta)^j}{j!})\sum_{s'\in S}p(s,s')Pmax[\dMA_{s'},\diamond^{\leq 0}G]\\
            \geq &(1-e^{-\lambda(s)\delta} - \frac{(\lambda(s)\delta)^{2N}}{(2N)!})\sum_{s'\in S}p(s,s')Pmax[\dMA_{s'},\diamond^{\leq 0}G]\\
            \geq &Pmax[\dMA_{s},\diamond^{\leq 1}G] -\frac{(\lambda(s)\delta)^{2N}}{(2N)!}\\
            \geq &Pmax[\dMA_{s},\diamond^{\leq 1}G] -\frac{(\lambda\delta)^{2N}}{(2N)!} = Pmax[\dMA_{s},\diamond^{\leq 1}G] -\epsilon_1^{2N}
        \end{align*}
\item Suppose that $s\in PS$, then:
        \begin{align*}
            &Pmax[\dMAlow_s,\diamond^{\leq 1}G] = \sup_{\pi\in Sched(\dMAlow)} \sum_{s'\in \MS}p_\pi(s,s')Pmax[\dMAlow_{s'},\diamond^{\leq 1}G]\\
            \geq&  \sup_{\pi\in Sched(\dMAlow)} \sum_{s'\in \MS}p_\pi(s,s')(Pmax[\dMA_{s'},\diamond^{\leq 1}G] - \epsilon_1^{2N})\\
            =&_{\textit{Non-zenoness assumption}}  \sup_{\pi\in Sched(\dMAlow)} \sum_{s'\in \MS}p_\pi(s,s')Pmax[\dMA_{s'},\diamond^{\leq 1}G] - \epsilon_1^{2N}\\
		=&_{\substack{\textit{available actions \&}\\ \textit{graph structure are equal}}}  \sup_{\pi\in Sched(\dMA)} \sum_{s'\in \MS}p_\pi(s,s')Pmax[\dMA_{s'},\diamond^{\leq 1}G] - \epsilon_1^{2N}\\
            = &Pmax[\dMA_s,\diamond^{\leq 1}G] - \epsilon_1^{2N}
        \end{align*}
\end{itemize}
	Now for the step case, we consider $k+1$ steps:
\begin{itemize} 
	\item Suppose that $s\in MS$:
            \begin{align*}
            &Pmax[\dMAlow_s,\diamond^{\leq k+1}G] \\
            = & (1-\sum_{j=0}^{2N}\frac{(-\lambda(s)\delta)^j}{j!})\sum_{s'\in S}p(s,s')Pmax[\dMAlow_{s'},\diamond^{\leq k}G] \\
		    &\hspace{3em} + \sum_{j=0}^{2N}\frac{(-\lambda(s)\delta)^j}{j!}Pmax[\dMAlow_s,\diamond^{\leq k}G] \\
            \geq &(1-e^{-\lambda(s)\delta}-\frac{(\lambda(s)\delta)^{2N}}{(2N)!})\sum_{s'\in S}p(s,s')Pmax[\dMAlow_{s'},\diamond^{\leq k}G] \\
		    &\hspace{3em} + e^{-\lambda(s)\delta}Pmax[\dMAlow_s,\diamond^{\leq k}G]\\
            \geq&_{\textit{IH}}(1-e^{-\lambda(s)\delta}-\frac{(\lambda(s)\delta)^{2N}}{(2N)!})\sum_{s'\in S}p(s,s')(Pmax[\dMA_{s'},\diamond^{\leq k}G]-\epsilon_{k}^{2N}) \\
		    &\hspace{3em} + e^{-\lambda(s)\delta}(Pmax[\dMA_s,\diamond^{\leq k}G] - \epsilon_k^{2N})\\
            =&(1-e^{-\lambda(s)\delta})\sum_{s'\in S}p(s,s')(Pmax[\dMA_{s'},\diamond^{\leq k}G]) + e^{-\lambda(s)\delta}(Pmax[\dMA_s,\diamond^{\leq k}G]) \\
		    & \hspace{3em}- \frac{(\lambda(s)\delta)^{2N}}{(2N)!}\sum_{s'\in S}p(s,s')Pmax[\dMA_{s'},\diamond^{\leq k}G]\\
		    & \hspace{3em}-(1-e^{-\lambda(s)\delta}-\frac{(\lambda(s)\delta)^{2N}}{(2N)!})\sum_{s'\in S}p(s,s')\epsilon_{k}^{2N} - e^{-\lambda(s)\delta} \epsilon_k^{2N}\\
            =&Pmax(\dMA_s,\diamond^{k+1}G) - (1-\frac{(\lambda(s)\delta)^{2N}}{(2N)!})\epsilon_k^{2N} - \frac{(\lambda(s)\delta)^{2N}}{(2N)!}\sum_{s'\in S}p(s,s')Pmax[\dMA_{s'},\diamond^{\leq k}G]\\
            \geq& Pmax(\dMA_s,\diamond^{k+1}G) - (1-\frac{(\lambda\delta)^{2N}}{(2N)!})\epsilon_k^{2N} - \frac{(\lambda\delta)^{2N}}{(2N)!}\\
            \geq & Pmax(\dMA_s,\diamond^{k+1}G) - (1+\frac{(\lambda\delta)^{2N}}{(2N)!})\epsilon_k^{2N} - \frac{(\lambda\delta)^{2N}}{(2N)!}
            \end{align*}
            
            \noindent Here, the last inequality could be left out to obtain slightly sharper bounds. However, keeping this inequality yields the same bounds for both $\dMAhigh$ and $\dMAlow$, so for sake of presentation, we keep them as is.
		We then have
            \begin{align*}
            & Pmax(\dMA_s,\diamond^{k+1}G) - (1+\frac{(\lambda\delta)^{2N}}{(2N)!})\epsilon_k^{2N} - \frac{(\lambda\delta)^{2N}}{(2N)!}\\
		    =& Pmax(\dMA_s,\diamond^{k+1}G) - (1+\frac{(\lambda\delta)^{2N}}{(2N)!})((1+\frac{(\lambda\delta)^{2N}}{(2N)!})^k-1) - \frac{(\lambda\delta)^{2N}}{(2N)!}\\
		    =& Pmax(\dMA_s,\diamond^{k+1}G) - ((1+\frac{(\lambda\delta)^{2N}}{(2N)!})^{k+1}-(1+\frac{(\lambda\delta)^{2N}}{(2N)!})) - \frac{(\lambda\delta)^{2N}}{(2N)!}\\
		    =& Pmax(\dMA_s,\diamond^{k+1}G) - ((1+\frac{(\lambda\delta)^{2N}}{(2N)!})^{k+1}-1)+\frac{(\lambda\delta)^{2N}}{(2N)!} - \frac{(\lambda\delta)^{2N}}{(2N)!}\\
		    =& Pmax(\dMA_s,\diamond^{k+1}G) - \epsilon_k^{2N}
            \end{align*}
    \item If $s\in PS$, the reasoning is analogous to the base case.
\end{itemize}
            This concludes this part of the proof, a similar methodology can be applied for $\dMAhigh$, where instead we get a positive sign for $\epsilon_k^{2N+1}$ to obtain an upper bound. Finally, the following recurrence relation has been found: \begin{align*}
                \epsilon_{n+1}^{2N} &= (1-\frac{(\lambda\delta)^{2N}}{(2N)!})\epsilon_n^{2N} + \frac{(\lambda\delta)^{2N}}{(2N)!}\\
                \epsilon_1^{2N} &= \frac{(\lambda\delta)^{2N}}{(2N!)}
            \end{align*}
            Solving this yields $\epsilon_n^{2N} = \frac{(\lambda\delta)^{2N}}{(2N!)} \sum_{j=0}^{n-1}(1+\frac{(\lambda\delta)^{2N}}{(2N)!})^j$. For $n=t/\delta$ we obtain:
            \begin{align*}
                \epsilon_{t/\delta}^{2N} =& \frac{(\lambda\delta)^{2N}}{(2N!)} \sum_{j=0}^{t/\delta-1}(1+\frac{(\lambda\delta)^{2N}}{(2N)!})^j\\
                =& ((1+\frac{(\lambda\delta)^{2N}}{(2N)!})^{t/\delta-1} + \frac{(2N)!}{(\lambda\delta)^{2N}}((1+\frac{(\lambda\delta)^{2N}}{(2N)!})^{t/\delta-1}-1))\frac{(\lambda\delta)^{2N}}{(2N)!}\\
                =&((1+\frac{(\lambda\delta)^{2N}}{(2N)!})^{t/\delta-1}(1+\frac{(2N)!}{(\lambda\delta)^{2N}})-\frac{(2N)!}{(\lambda\delta)^{2N}})\frac{(\lambda\delta)^{2N}}{(2N)!}\\
                =& (1+\frac{(\lambda\delta)^{2N}}{(2N)!})^{t/\delta}-1
            \end{align*}
            Therefore, we obtained the error bounds as claimed. The case for $Pmin(\cdot)$ is analogous.
        \qed
\end{proof}
Theorem~\ref{thm:approxDiscretizationError} then follows as a direct corollary of Lemmas~\ref{lemma:pMAapproxApproximations} and \ref{lemma:approxDiscretizationBounds}.

\section{Details on correct use of Parameter Lifting}\label{appendix:extremaPL}

As covered in Section~\ref{ss:ParameterLifting}, we need additional global assumptions to ensure that the transformations throughout Section~\ref{section:results} yield parametric formulas such that the extrema are attained at one of the vertices of the region. To reiterate, over $\PS$ transitions, we assumed the following:

``I) To ensure that the elimination of $\PS$ chains yields monotonic parametric functions, we assume that no two states of such a chain have a parametric function over the same variable and II) to cover the collapse of $\PS$ in $\MS$, we assume that the sets of parameters over rate- and probabilistic transition functions are disjoint.''

\noindent Therefore, we have to show two things: I) by adding additional auxiliary variables, we increase the number of dimensions. Due to this, it could be the case that one partitions the region up to arbitrary precision, whilst the relevant hyperplane remains undecided. However, we show that this is not the case and one can still decide this hyperplane up to arbitrary precision. II) If we can show that the probabilistic transitions of the approximate discretization are indeed monotone, we have shown that extrema are attained at the vertices of the region, due to the structure of the solution function of the relaxed MDP as given in~\cite{DBLP:journals/fmsd/JungesAHJKQV24}. In Lemma~\ref{lemma:monotonicity}, we show that this property is maintained throughout discretization. 

First of all, we start with showing I). In particular, we show that the solution function of the pMA is continuous over each parameter, by using a known result on continuity of graph-preserving parametric Markov chains~\cite{DBLP:journals/fmsd/JungesAHJKQV24}. Thus, if one can decide whether $\MA[v]$ is accepting or not by considering an interval on the hyperplane, one can find a small enough subregion on the higher dimensional model, containing $v$, such that the same conclusion can be drawn.

\begin{theorem}\label{thm:auxiliaryVars}
    Let $\MA$ be a pMA over parameters $\mathcal{X}=\{x_1, ..., x_n\}$, region $\Region$, $t\in \mathbb{R}_{>0}$ and $G\subseteq S$. Consider the property $Pmax[\MA[v], \diamond^{\leq t}G]$ and an oracle which, for every $D\subseteq \Region$, gives lower and upper bounds for this probability with non-zero error $\epsilon$, where $\epsilon\rightarrow 0$ as $|D|\rightarrow 0$.
    
    Now, consider pMA $\MA'$ over parameters $\mathcal{X}' = \{x_1,...,x_n,x_i'\}$, obtained by replacing some occurrences of parameter $x_i$ by auxiliary parameter $x'_i$, and region $\Region' = \Region\times \Region|_{x_i}$. If one can decide a valuation $(v_1,...,v_n)\in \Region$ with this oracle, there exists an $\eta$-ball in $\Region'$ such that the valuation can be decided by our approach on a hyperrectangle $D'= \bigotimes_{i\in [|\mathcal{X}'|]}[a_i,b_i], a_i<b_i$ contained in this $\eta$-ball and $(v_1, ..., v_n, v_i)\in D'$. 
\end{theorem}
\begin{proof}
    Let $D\subseteq \Region$ with $v\in D$ and let $lb,ub$ be the bounds obtained for checking region $D$ with this oracle. Suppose that one can decide satisfaction or violation of the property by the model induced by valuation $v$ from these bounds. Then, per our non-zero assumption of the oracle, $\epsilon \leq \min(Pmax[\MA[v], \diamond^{\leq t}G] - lb, ub -Pmax[\MA[v], \diamond^{\leq t}G])$. Now we will show that there exists an $\eta>0$ such that for any $w$ in the $\eta$-ball around $v'=(v_1, ..., v_n, v_i)$, we have that $|Pmax[\MA'[v'], \diamond^{\leq t}G] - Pmax[\MA'[w], \diamond^{\leq t}G]| < \epsilon$. That is, the solution function is continuous with respect to the parameters of the pMA, from which we can conclude that there exists some hyperrectangular region with which to decide acceptance/violation of $v'$. Thus, to show continuity:
    \begin{align*}
        &|Pmax[\MA'[v'], \diamond^{\leq t}G] - Pmax[\MA'[w], \diamond^{\leq t}G]|\\
        \leq&_{\textit{for small enough $\delta$}} |Pmax[\dMA[v'], \diamond^{\leq t}G] - Pmax[\dMA[w], \diamond^{\leq t}G]| + \frac{2}{3}\epsilon\\
    \end{align*}
    It is known that graph preserving pMC give continuous solution functions~\cite{DBLP:journals/fmsd/JungesAHJKQV24}. Thus, as we can rewrite $Pmax[\dMA[v'], \diamond^{\leq t}G]$ as the maximum of all reachability probabilities of a DTMC induced by a non-randomized scheduler, of which there are finitely many, and the maximum of finitely many continuous functions is continuous as well, we get that $Pmax[\dMA, \diamond^{\leq t}G]$ admits a continuous solution function with respect to its parameters. Thus there exists an $\eta>0$ such that:
    \begin{align*}
        |v-w|<\eta \rightarrow |Pmax[\dMA[v], \diamond^{\leq t}G] - Pmax[\dMA[w], \diamond^{\leq t}G]| < \frac{\epsilon}{3}
    \end{align*}
 Thus, we conclude that, for any valuation $w$ in the $\eta$-ball centered around $v'$, we obtain 
    \begin{align*}
        |Pmax[\MA'[v'], \diamond^{\leq t}G] - Pmax[\MA'[w], \diamond^{\leq t}G]| < \epsilon
    \end{align*}
    In particular, this means that there exists a hyperrectangle $D'$, contained in this ball, for which our approach yields less strict bounds than obtained by the oracle. Therefore, we can decide $v'$ using our approach and the proof is concluded as $Pmax[\MA'[v'], \diamond^{\leq t}G] = Pmax[\MA[v], \diamond^{\leq t}G]$ per construction of $\MA'$.\qed
\end{proof}

Note that the theorem above does not require our approach to work on the model without auxiliary variables, but our approach can decide at least as much as any oracle yielding results with non-zero error bounds. As many variables may be substituted, iterative application of Theorem~\ref{thm:auxiliaryVars} yields the same claim of our result.

Secondly, we show monotonicity of the discretized model. Here, it is important to note that the collapse of states does not impact the result, as they can safely assumed to be over different parameters.

\begin{lemma}\label{lemma:monotonicity}
    Let $\MA$ be a pMA, $R\subseteq \mathbb{R}^{|\mathcal{X}|}$ a hyperrectangular region, $s\in \MS$ and $\delta >0$ a discretization step according to Theorem~\ref{thm:GuckDiscretization}. Suppose that $R(s,s')$ is multi-affine for all $s'\in S$, then the extrema of the following function is attained at one of the vertices of $R$:
    \begin{align*}
        &f_1(x) = (1-e^{-\lambda(s)(x)\delta}) p(s,s')(x)
    \end{align*}
\end{lemma}
\begin{proof}
	As $R(s,s')(x)$ is multi-affine for all $s' \in S$, so is $\lambda(s) = \sum_{s'\in S}R(s,s')$. Moreover, as mentioned in Section~\ref{section:prelims}, multi-affine functions are monotone in each variable. It suffices to show that $f_1(x)$ is monotone. Therefore, consider $\frac{d}{dx}f_1$:
    \begin{align*}
        \frac{d}{dx}f_1(x) &= 
	    \frac{d}{dx} \left((1-e^{-\lambda(s)(x)\delta}) \frac{R(s,s')(x)}{\lambda(s)(x)}\right)\\
		&= \frac{\lambda(s)(x)\frac{d}{dx}R(s,s')(x) - R(s,s')(x)\frac{d}{dx}\lambda(s)(x)}{(\lambda(s)(x))^2}(1-e^{-\delta\lambda(s)(x)}) \\
		&\hspace{3em}+ \delta\frac{d}{dx}[\lambda(s)(x)]e^{-\delta\lambda(s)(x)}(\frac{R(s,s')(x)}{\lambda(s)(x)})
    \end{align*}
    Now we denote $\phi(x) = \frac{\lambda(s)(x)\frac{d}{dx}R(s,s')(x) - R(s,s')(x)\frac{d}{dx}\lambda(s)(x)}{(\lambda(s)(x))^2}$. From now on, assume that $\lambda$ is monotone increasing in $x$, where the decreasing case is analogous.

    In case that $\phi(x)\geq 0 $, we get:
    \begin{align*}
        \frac{d}{dx}f_1(x) \geq \phi(x)(1-e^{-\delta\lambda(s)(x)}) \geq \phi(x)(1-e^{-\delta\inf_{z\in R}\lambda(s)(z)}) \geq 0
    \end{align*}
    Inequality $\inf_{z\in R}\lambda(s)(z) >0$ holds due to the assumption that region valuations are well-defined, as made in Section~\ref{section:prelims}.

    Similarly, for $\phi \leq 0$:
    \begin{align*}
	    \frac{d}{dx}f_1(x) &= \phi(x)(1-e^{-\delta\lambda(s)(x)}) + \delta\frac{d}{dx}[\lambda(s) (x)]e^{-\delta\lambda(s)(x)}(\frac{R(s,s')(x)}{\lambda(s)(x)}) \\
	    &= e^{-\delta\lambda(s)(x)} \left(\phi(x)e^{\delta\lambda(s)(x)} - \phi(x) + \delta\frac{d}{dx}[\lambda(s) (x)](\frac{R(s,s')(x)}{\lambda(s)(x)})\right) \geq 0
    \end{align*}
	This holds, as for small enough $\delta$, we have that:
    \begin{align*}
	    \delta\frac{d}{dx}[\lambda(s)(x)](\frac{R(s,s')(x)}{\lambda(s)(x)})\geq \phi(x)(1-e^{\delta\lambda(s)(x)})
    \end{align*}
	Thereby concluding the proof. \qed
\end{proof}

In this Lemma~\ref{lemma:monotonicity}, we do not cover the case of discretizing a self-loop, as each MA can be transformed to an equivalent MA without self-loops in Markovian states. Therefore, the self-loop probability of a state $s$ in the discretized model is monotone increasing/decreasing if $\lambda(s)$ is respectively decreasing/increasing.

\fi

\end{document}